\documentclass[11pt]{article}
\usepackage{latexsym}
\usepackage{amsmath,amssymb,amsthm}
\usepackage{epsfig}
\usepackage[right=0.8in, top=1in, bottom=1in, left=0.8in]{geometry}
\usepackage{setspace}
\usepackage{bbold}
\usepackage{mathtools}
\usepackage{dirtytalk}
\usepackage{enumerate}
\usepackage{color,soul}
\usepackage{stackengine}

\newcommand{\zo}{\{0,1\}}

\newcommand{\N}{{\mathbf{N}}}

\newcommand{\R}{{\mathbb{R}}}

\newcommand{\eps}{\epsilon}

\newcommand{\polylog}{\mathop\mathrm{polylog}\nolimits}
\newcommand{\xor}{\oplus}

\newcommand{\UC}{{\mathcal{U}}}

\newcommand\independent{\protect\mathpalette{\protect\independenT}{\perp}}
\def\independenT#1#2{\mathrel{\rlap{$#1#2$}\mkern2mu{#1#2}}}

\newtheorem{theorem}{Theorem}[section]
\newtheorem*{theorem*}{Theorem}

\newtheorem{corollary}[theorem]{Corollary}
\newtheorem{lemma}[theorem]{Lemma}

\newtheorem{definition}[theorem]{Definition}
\newtheorem{claim}[theorem]{Claim}
\newtheorem{fact}[theorem]{Fact}

\newtheorem{remark}[theorem]{Remark}

\makeatletter
\newtheorem*{rep@theorem}{\rep@title}
\newcommand{\newreptheorem}[2]{%
\newenvironment{rep#1}[1]{%
 \def\rep@title{#2 \ref{##1}}%
 \begin{rep@theorem}}%
 {\end{rep@theorem}}}
\makeatother
\newreptheorem{theorem}{Theorem}

\newcommand{\E}{{\mathbb E}}

\newcommand{\norm}[1]{\left\lVert#1\right\rVert}
\newcommand{\normo}[1]{\norm{#1}_1}
\newcommand{\normt}[1]{\norm{#1}_2}

\usepackage[pdftex,dvipsnames]{xcolor}  %
\usepackage[utf8]{inputenc}
\usepackage{amsmath}
\usepackage{hyperref}
\usepackage{cryptocode}   
\graphicspath{ {img/} }
\usepackage[colorinlistoftodos,prependcaption,textsize=tiny]{todonotes}
\newcommandx{\unsure}[2][1=]{\todo[linecolor=red,backgroundcolor=red!25,bordercolor=red,#1]{#2}}
\newcommandx{\change}[2][1=]{\todo[linecolor=blue,backgroundcolor=blue!25,bordercolor=blue,#1]{#2}}
\newcommandx{\info}[2][1=]{\todo[linecolor=OliveGreen,backgroundcolor=OliveGreen!25,bordercolor=OliveGreen,#1]{#2}}
\newcommandx{\improvement}[2][1=]{\todo[linecolor=Plum,backgroundcolor=Plum!25,bordercolor=Plum,#1]{#2}}
\newcommandx{\thiswillnotshow}[2][1=]{\todo[disable,#1]{#2}}

\DeclareMathOperator{\CT2p}{\textsc{2pCT}}
\DeclareMathOperator{\IT2p}{\textsc{2pIT}}
\DeclareMathOperator{\Poi}{Poi}
\DeclareMathOperator{\Mult}{Mult_{-0}}

\def\poly{\operatorname{poly}}
\def\polylog{\operatorname{polylog}}

\renewcommand{\bf}{\normalfont \bfseries}
\renewcommand{\sc}{\normalfont \scshape}

\renewcommand{\paragraph}[1]{\vspace{0.7ex}\noindent{\bf #1}}

\title{
Two Party Distribution Testing:\\
Communication and Security
}
\author{Alexandr Andoni\\Columbia University\\andoni@cs.columbia.edu 
\and Tal Malkin\\Columbia University\\tal@cs.columbia.edu 
\and Negev Shekel Nosatzki\\Columbia University\\ns3049@columbia.edu}

\begin{document}

\maketitle

\begin{abstract}
We study the problem of discrete distribution testing in the {\em
  two-party setting}.  For example, in the standard closeness testing
problem, Alice and Bob each have $t$ samples from, respectively,
distributions $a$ and $b$ over $[n]$, and they need to test whether
$a=b$ or $a,b$ are $\eps$-far (in the $\ell_1$ distance) for some
fixed $\eps>0$. This is in contrast to the well-studied one-party
case, where the tester has unrestricted access to samples of both
distributions, for which optimal bounds are known for a number of
variations. Despite being a natural constraint in applications, the
two-party setting has evaded attention so far.

We address two fundamental aspects of the two-party setting: 1) what
is the communication complexity, and 2) can it be accomplished
securely, without Alice and Bob learning extra information about each
other's input. Besides closeness testing, we also study the
independence testing problem, where Alice and Bob have $t$ samples from
distributions $a$ and $b$ respectively, which may be correlated; the
question is whether $a,b$ are independent of $\eps$-far from being independent.
Our contribution is three-fold:
\begin{itemize}
\item
{\bf Communication}: we show how to gain communication efficiency as
we have more samples, beyond the information-theoretic bound on
$t$. Furthermore, the gain is polynomially better than what one may
obtain by adapting one-party algorithms.

For the closeness testing, our protocol has communication $s =
\tilde{\Theta}_{\eps}\left(n^2/t^2\right)$ as long as $t$ is at least
the information-theoretic minimum number of samples. For the
independence testing over domain $[n] \times [m]$, where $n\ge m$, we obtain $s =
\tilde{O}_{\eps}(n^2 m/t^2 + n m/t + \sqrt{m})$.
\item
{\bf Lower bounds}: we prove {\em tightness} of our trade-off for the
closeness testing, as well as that the independence testing requires
tight $\Omega(\sqrt{m})$ communication for unbounded number of
samples. These lower bounds are of independent interest as, to the
best of our knowledge, these are the first 2-party communication lower
bounds for testing problems, where the inputs represent a set of {\em
  i.i.d.~samples}.
\item
{\bf Security}: we define the concept of secure distribution testing
and argue that it must leak at least some minimal
information when the promise is not satisfied. 
We then provide secure versions of the above protocols with an
overhead that is only polynomial in the security parameter.
\end{itemize}

\end{abstract}

\thispagestyle{empty}
\newpage
\setcounter{page}{1}

\section{Introduction}

Distribution property testing is a sub-area of statistical hypothesis
testing, which has enjoyed continuously growing interest in the
theoretical computer science community, especially since the 2000
papers \cite{goldreich2000testing,batu2000testing}.
One of the most basic problems is
closeness testing, also known as the {\em homogeneity testing} or {\em
  two sample problem}; see \cite{gutman1989asymptotically,
  shayevitz2011renyi, unnikrishnan2012optimal}. Here, given two
distributions $a,b$ and $t$ samples from each of them, distinguish
between the cases where $a=b$ versus $a$ and $b$ are $\eps$-far, which
usually means $\|a-b\|_1> \eps$.\footnote{This is equivalent to saying
  that the total variation distance is more than $\eps/2$.}  For this
specific problem, the extensive research lead to algorithms with
optimal sample complexity \cite{batu2000testing, valiant2011testing,
  batu2013testing, chan2014optimal, diakonikolas2016new,
  diakonikolas2016collision}, including when the number of samples
from the two distributions is unequal \cite{acharya2014sublinear,
  bhattacharya2015testing, diakonikolas2016new}. Further research
directions of interest include obtaining instance-optimal algorithms,
which depend on further properties of the distributions $a,b$
\cite{acharya2011competitive, acharya2012competitive,
  diakonikolas2016new}, quantum algorithms \cite{bravyi2011quantum},
as well as algorithms whose output is differentially-private
\cite{diakonikolas2015differentially,cai2017priv,
  DBLP:journals/corr/AcharyaSZ17,
  DBLP:journals/corr/AliakbarpourDR17}%
. An even larger body of work
studied numerous related problems such as independence testing,
among many others.
We refer the reader to the surveys \cite{goldreich17introduction,
  canonne2015survey, rubinfeld2012taming,rubinfeld2006sublinear} for
further references.

Focusing on testing two distributions, such as in the closeness
problem, a very natural aspect has, surprisingly, evaded attention so
far: that such a testing task would be often run by two players, each
with access to their own distribution.  Specifically, Alice has
samples from the distribution $a$, Bob has samples from the
distribution $b$, and they need to jointly solve a distribution
testing problem on $(a,b)$. This is a natural setting that models many
of the envisioned usage scenarios of distribution testing, where
different parties wish to jointly perform a statistical hypothesis
testing task on their distributions. For example,
\cite{shayevitz2011renyi} describes the scenario where two distinct
sensors need to test whether they sample from the same distribution
(``noise'') or not.

This 2-party setting raises the following standard theoretical
challenges, neither of which has been previously studied in the
context of distribution testing:
\begin{itemize}
\item
What is the {\em communication complexity} of the testing problem? In
particular, can we do better than the straightforward
approach, where Alice sends her samples to Bob who then runs an offline
algorithm? Can we prove matching communication lower bounds for such testing problems?

This aspect parallels the quest for low memory or communication usage
for hypothesis testing on a {\em single distribution}, initiated in
the statistics community \cite{cover1969hypothesis,
  hellman1970learning} and \cite{ahlswede1986hypothesis,
  han1987hypothesis, amari1998statistical}.
\item
Is it possible to design a distribution testing protocol that is {\em secure}, i.e., where
Alice and Bob do not learn anything about each other's samples,
besides the mere fact of whether their distributions are same or
$\eps$-far?

While more modern, this question is highly relevant in today's push
for doing statistics that is more privacy-respecting.
\end{itemize}

\subsection{Our Contributions}
In this paper, we initiate the study of testing problems in the
two-party model, and design protocols which are both
communication-efficient and secure. We do so for two basic problems on
pairs of distributions (i.e., where the two-party setting is natural):
1) closeness testing, and 2) independence testing.

Our main finding is that, once the number of samples exceeds the
information-theoretic minimum, we can obtain protocols with
polynomially smaller communication than the na\"ive adaptation of
existing algorithms.  We complement our protocols with lower bounds on
the communication complexity of such problems that are near-optimal
for closeness testing, as well as for independence in an intriguing
parameter regime. Our upper and lower bounds on communication are
novel even without any security considerations.

To argue security, we also put forth a definition for secure
 distribution testing in the multi-party model. Our definition
differs from the standard secure computation setting due to two unique
features of the considered setting. First, this is ``testing'' and not
``computing''; second, the function of interest is defined with
respect to distributions, but the inputs that the parties use in the
computation are samples.  These features do not come into play if the
distributions satisfy the promise (e.g., they are either identical or
$\epsilon$-far), in which case the security guarantee matches the
standard cryptographic one (no information is leaked beyond the
output).  However, the crux is when the promise is not satisfied, in
which case we need to allow for some information on the parties'
samples to be leaked by the protocol.  Our definition permits leakage
of at most one bit in this case, and leaks nothing when the promise is
satisfied. See the formal definition and discussion in
Section~\ref{sec:secureCT}.

\smallskip
\paragraph{Closeness Testing.} 
In the {\em 2-party closeness testing} problem $\CT2p_{n,t,\epsilon}$, Alice and Bob each have access to $t$ samples from some
distributions respectively $a, b$ over alphabet $[n]$. Their goal is
to distinguish between $a=b$ and $\normo{a-b}\geq\epsilon$ with
probability $\ge2/3$.

We first give a non-secure near-optimal communication protocol, and
then show how to make it secure with only a small overhead (polynomial
in the security parameter).  Our secure version is based on the
existence of a PRG that stretches from $\polylog(m)$ bits to $m$ bits,
and of an OT protocol with $\polylog$ communication. We elaborate on
these standard cryptographic assumptions in a later section.  Overall,
we prove the following theorem.

\begin{theorem*}[Closeness, Secure; see Theorem~\ref{priv_cl_tst}]
Fix a security parameter $k>1$. Fix $n>1$ and $\eps\in(0,2)$, and let $t$ be
such that $t\ge C \cdot k \cdot \max\left(n^{2/3}\cdot
\epsilon^{-4/3}, \sqrt{n}\cdot \epsilon^{-2}\right)$ for some (universal) constant
$C>0$. Then, assuming PRG and OT as above, there exists a {\em secure
  distribution testing protocol} for $\CT2p_{n,t,\epsilon}$ which uses
$\tilde{O}_k\left(\frac{n^2}{t^2\epsilon^4} + 1\right)$ communication.
\end{theorem*}

To contrast the communication bounds of our protocol to the classic
1-party setting, consider what happens in the extreme settings of the
parameters $s,t$, for a fixed $\eps$. When $t\approx \Theta(n^{2/3})$,
the communication is $\tilde O(n^{2/3})$ as well, i.e., Alice may as
well just send all the samples over to Bob. However the communication
decreases as the players have more samples. This may not be surprising
given the testing results with unequal number of samples
\cite{bhattacharya2015testing,diakonikolas2016new}: indeed, Alice can
send $\approx \max\{n/\sqrt{t},\sqrt{n}\}$ samples to Bob, and Bob can
run the tester. In contrast, our protocol obtains a {\em polynomially
  smaller} complexity, $\approx n^2/t^2$, whenever $t\gg
n^{2/3}$. Intuitively, considering the extreme of $t\gg n$, we
can obtain near-constant communication: with so many samples, we can
{\em learn} the distribution, and then use {\em sketching} tools, such
as the $\ell_1$ sketching algorithm of \cite{AMS, I00b}.

We prove a {\em near-tight} lower bound on the above communication
complexity trade-off (even without security considerations) in
Section~\ref{sec:ctLB}. We note that this lower bound differs from
standard communication complexity lower bounds as the players' inputs
are {\em i.i.d.~samples} and not worst-case.

\begin{theorem*}[Closeness lower bound; see Theorem~\ref{thm:closeDisj}]
Any two-way communication protocol for $\CT2p_{n,t,1/2}$ requires
$s=\tilde\Omega\left(\frac{n^2}{t^2}\right)$ communication.
\end{theorem*}

\paragraph{Independence Testing.}
The second problem we consider is the independence testing problem in the
2-party model, denoted $\IT2p_{n,m,t,\epsilon}$. Let $p=(a,b)$ be some
joint distribution over $[n] \times [m]$, where $n\ge m$, and for $i
\in [t]$, let $\zeta_i$ be %
a sample drawn from $p$. Now we
provide Alice with the first coordinates of 
$\zeta_i$'s
and Bob with the second coordinates. Alice and Bob's goal is to test whether %
$p$ is a product distribution or $\epsilon$-far from any product
distribution.
We prove the following: 

\begin{theorem*}[Independence, Secure; see Theorem~\ref{thm_2pit}] %
Fix a security parameter $k>1$. Fix $\eps\in(0,2)$, $1\le m\le n$, and let $t$ be such that $t
\geq C \cdot k \cdot \left(n^{2/3}m^{1/3}\epsilon^{-4/3} +
\sqrt{nm}/\epsilon^2\right)$, for some (universal) constant $C$, and assuming OT, there is a secure
distribution testing protocol for $\IT2p_{n,m,t,\epsilon}$ using
$\tilde{O}_{k}\left(\frac{n^2\cdot m}{t^2\epsilon^4} + \frac{n\cdot
  m}{t \epsilon^4} + \frac{\sqrt{m}}{\epsilon^3}\right)$ bits of communication.
\end{theorem*}

We note that the lower bound on $t$ from the above theorem is
necessary as it is the information-theoretic bound, as proven
in~\cite{diakonikolas2016new}.

An important qualitative aspect of the communication complexity for
$\IT2p$ is that, when the number of samples $t\to \infty$, the
protocol uses $\tilde\Theta_\epsilon(\sqrt{m})$ bits
communication. This is in contrast to $\CT2p$, where the
communication becomes poly-logarithmic for $t\to \infty$. Indeed, we
show that $\Omega(\sqrt{m})$ communication is necessary for one-way
communication for $\IT2p$. Since our non-secure protocol can be made
one-way (see Remark~\ref{rem:it1way}), this lower bound is tight for
one-way protocols. We conjecture that the bound on $s$ from the
above theorem is near-tight in $n,m,t$ for 2-way communication protocols, even without security.

\begin{theorem*}[Independence lower bound; see Theorem~\ref{thm_lb_it}]
For $n, t \in \N$, any one-way protocol for $\IT2p_{n,n,t,1}$
requires $\Omega(\sqrt{n})$ bits of communication.
\end{theorem*}

\subsection{Related work}
Our work bridges three separate areas and models: distribution testing,
streaming/sketching, and secure computation.  There's a large body of
work in each of these areas, and we mention the branches most relevant
to us.

\paragraph{Distribution Testing.} The most basic problem for
distribution testing is identity testing: where one distribution
(say, $a$) is fixed/universal (e.g., uniform). There is a large body of
work on this problem for general $a$, starting with
\cite{Batu:2001:TRV:874063.875541}; see surveys
\cite{goldreich17introduction,
  canonne2015survey, rubinfeld2012taming}. As this
problem is fundamentally about one unknown distribution $a$, we do not
consider it in the 2-party model.

For the closeness testing problem, in addition to the aforementioned
references, we highlight the result of \cite{diakonikolas2016new} (see
also \cite{goldreich2016uniform}), whose techniques are the starting
point of our protocols. That result introduces a clean framework for
testing, by reducing $\ell_1$ testing to $\ell_2$ testing and checking
if two distributions are the same or $\epsilon$-far from each other
with respect to the $\ell_2$ norm. 

The independence testing problem has been studied in
\cite{Batu:2001:TRV:874063.875541, levi2013testing,
  acharya2015optimal, diakonikolas2016new}. In the standard setting (1
party) the problem is defined as: given samples from a distribution
$\pi$ over pairs with marginals $a,b$, determine whether $a,b$ are
independent ($\pi=a\times b$) or $\pi$ is $\eps$-far from being a
product distribution. This line of work culminated in the work of
\cite{diakonikolas2016new}, who show the tight sample complexity of
$\Theta(\sqrt{nm}/\eps^2 + n^{2/3}m^{1/3}/\eps^{4/3})$, where $n\ge
m$ are the cardinalities of the two sets.

Other related questions include: testing over $d$-tuples of
distributions \cite{levi2013testing, DBLP:journals/siamdm/LeviRR14,
  acharya2015optimal, diakonikolas2016new}, testing independence of
monotone distributions \cite{batu2004sublinear}, and testing $k$-wise
independence \cite{AAKMRX, rubinfeld2010testing,
  rubinfeld2013robust}.
To what extent the samples--communication trade-off from this paper
extend to these problems is an interesting open question. 

\paragraph{Connection to communication and sketching.}
As already alluded to, our 2-party testing setup turns out to be connected to
sketching. In particular, for the closeness problem, when the
number of samples is $t\approx n$, Alice and Bob can approximate
respectively $a,b$ (up to small $\ell_1$ distance), and then just
estimate $\|a-b\|_1$ using $\ell_1$ sketching such as \cite{I00b},
using {\em constant} communication.

More generally, the 2-party communication model is intimately connected to
the {\em streaming model}, where the input is streamed over, while
keeping small extra space for computation. Most relevant here is the
work where the input is randomized in some way --- this is similar to our
setting, where the input consists of {\em samples} from a distribution. This
includes the recent work on {\em stochastic streaming}
\cite{crouch2016stochastic}, where the input is generated from a
distribution. In particular, \cite{crouch2016stochastic} analyze the
space complexity for estimating the $\ell_2$ norm of a probability
distribution $a$ given $t$ samples presented as a stream. They show a
trade-off between space $s$ and number of samples $t\ge\sqrt{n}$ of
$s\cdot t=\tilde O_\eps(n)$. They also show this is tight under a likely
conjecture on the space complexity of estimating frequency moments in
the {\em random order} streaming model. 

For the independence problem, two variants were studied in the
streaming model in \cite{indyk2008declaring,
  braverman2010measuring}. In the first variant, termed {\em
  centralized model}, we stream over pairs $\zeta=(x,y)$, which
implicitly define a joint distribution $(a,b)$, and need to compute
the distance between the joint distribution and the
product-of-the-marginals distribution. Space complexity of this
version is $(\log n)^{O(1)}$ \cite{braverman2010measuring}. The second
variant, termed {\em distributed model}, is more similar to ours: we
stream over items $(t,i,p)$, where $t$ is the time (index) of a sample
$(x_t,y_t)\sim(a,b)$ and $i\in[n]$ is either $x_t$ (if $p=1$) or $y_t$
(if $p=2$), and the goal is again to estimate the distance between
$(a,b)$ and product distribution. For the latter problem,
\cite{indyk2008declaring} show the complexity is between $\Omega(n)$
and $\tilde O_{\eps}(n^2)$.

In the context of connections between communication complexity and
property testing, we also mention the work of \cite{blais2012property,
  blais2016alice}.  They use communication tools to prove lower bounds
on the {\em sample complexity} of testing problems. Their techniques
however do not readily apply to proving {\em communication complexity} of
(distribution) testing problems as we do here.

\paragraph{Testing and learning in distributed and streaming models.}
Communication lower bounds are also often related to those in the
streaming and distributed models, which have received recent focus in
the context of testing and learning questions. In particular,
streaming (as a memory constraint) was considered as early as in
\cite{cover1969hypothesis, hellman1970learning}, for distribution
(hypothesis) testing of one distribution. In recent years, a lot of
attention has been drawn to streaming (memory) lower bounds for
learning problems, such as parity learning
\cite{raz2016fast,raz2017time,kol2017time,moshkovitz2017mixing,garg2018extractor}. These
results also show a trade-off between number of samples and space
complexity.

Another recent avenue is to study such problems in the distributed
model, where there are many symmetric players, each with a number of
i.i.d.~samples from the same distribution. For learning problems
(e.g., parameter or density estimation), see, e.g.,
\cite{braverman2016communication, communicationLearning-nips17,
  DaganShamir-correlations18}. We note that since learning is a much
harder problem, typically proving lower bounds is easier (e.g., as
shown in \cite{communicationLearning-nips17}, merely communicating the
output requires $\Omega(n)$ communication). In contrast, for testing
problems, the output is just one bit. For testing problems, see the
recent (independent) manuscript
\cite{DBLP:journals/corr/abs-1804-06952}, who studies testing of one
distribution, and focuses on reducing the communication per sample (from
a max of $O(\log n)$).

None of the lower bounds from the above papers seem to be relevant
here as they become vacuous for a 2-party setting. Indeed, when two
players have two set of samples from the {\em same} distribution, then
purely doubling the sample set of a player trivializes the question
(she can solve it without communication).

\paragraph{Secure approximations.}  
Our results on secure distribution testing can also be seen in the
context of the area of secure computation of approximations.  This is
a framework  introduced by \cite{Feigenbaum:2006:SMC:1159892.1159900},
allowing to combine the benefits of approximation algorithms and
secure computation.  This was considered in different settings
\cite{Feigenbaum:2006:SMC:1159892.1159900,HKKN01,BCNW08,Indyk:2006:PPA:2180286.2180304,BHN09,IMSW09,KMSZ08},
but the most relevant to us is private approximation of distance
between two input vectors. In particular, for $\ell_2$ distance, 
Alice and Bob each have a vector $a,b\in \R^n$ and want to estimate 
$\|a-b\|_2$, without revealing any information that does not follow
from the $\ell_2$ distance itself.  
For this problem, \cite{Indyk:2006:PPA:2180286.2180304}
show that secure protocols are possible with only poly-logarithmic
communication complexity. They also specifically look into the private
near-neighbor problem and its approximation.  We use some of their 
techniques in our secure  protocols.

Approximation and testing have a similar flavor in that
they both trade accuracy for efficiency, in different ways
(computing an estimate of the solution in the former case, and
computing the correct output bit if the promise is satisfied in
the latter). The security goals are also similar (prevent leakage
beyond the intended output).
One important difference is that the intended output in secure testing
is just the  single bit of whether or not the test passed.  
Thus, for example, when approximating a distance function, even a
secure protocol can leak any information that follows from the
distance. In contrast, when testing for closeness, if the inputs are
either identical or far, the protocol may only reveal this fact, but
no other information about what the distance is. 

\paragraph{Security and privacy of testing.}
While we are not aware of any work on secure testing, several recent
papers address {\em differentially-private distribution testing}
\cite{diakonikolas2015differentially, cai2017priv,
  DBLP:journals/corr/AcharyaSZ17,DBLP:journals/corr/AliakbarpourDR17}.
Here the privacy guarantee relates to the value of the output after
the computation is concluded, requiring it to be
differentially-private with respect to the inputs.
Our notion of security for distribution testing is different, in the
same way that secure computation is different from differentially
private computation.  While differential privacy (DP) is concerned
with what the intended output may leak about the inputs (even if the
input came from a single party or the computation is done by a trusted
curator), secure 2-party computation is concerned with how to compute
an intended output without leaking any information beyond what the
output itself reveals. 
The difference in goals also results in the difference that DP testing
privacy guarantees are typically statistical in nature and provide a 
non-negligible adversarial advantage, while secure testing protocols
rely on cryptographic assumptions and provide negligible advantage.

More recently, a more stringent model of Locally Differentially
Private Testing was proposed
\cite{Sheffet2018LocallyPH,Acharya2018TestWT}. This model provides a
stronger notion of differential privacy, where users send noisy
samples to an untrusted curator, and the goal is to allow the curator
to test the distribution of user inputs (for some property) without
learning \say{too much} about the individual samples. For LDP, the
main goal is to optimize the sample complexity as a function of the
privacy guarantees.  While this notion of privacy also incorporates
some privacy of the individual inputs, it is much closer to DP than to
our security notion. In addition, both DP and LDP do not provide
sub-linear communication (in the sample size, as we achieve here). In
fact, their goal is to allow $O(1)$ communication per sample, with
minimal sample overhead. In contrast, our protocols provide security
\say{nearly for free} while allowing for faster communication with
more samples.
Finally, in the case of independence testing, our work assumes samples
are distributed between the parties who need to test the joint
distribution, while in the aforementioned work, each data point
contains full sample information.

We also mention the work of \cite{MNS11} on {\em sketching in adversarial
environments}, which considers another  threat model:  here
the two parties performing the computation are honest (and trust each
other), but their (large) input arrives in a streaming fashion chosen
adversarially, and they may only maintain a sketch of the inputs.

\subsection{Our Techniques}

We now outline the techniques used to establish our main
results. Since our overall contribution is painting a big picture of
the 2-party complexity of distribution testing problems, we appeal to
a number of diverse tools. First, we design communication-efficient
protocols whose communication improves as we have more samples.
Second, we show that some of the established trade-offs are
near-optimal by proving communication complexity lower bounds on the
considered problems, which are near-tight in some of the parameter
regimes.  Third, we show how to transform our protocols into secure
protocols, under standard cryptographic assumptions, without further
loss in efficiency. All three of these contributions are independently
first-of-a-kind, to the best of our knowledge.

\medskip
\paragraph{Communication-efficient protocols.}
We start by noting that we can reduce the testing problem under the
$\ell_1$ distance (total variation distance) to the same problem under
the $\ell_2$ distance, using now-standard methods of
\cite{diakonikolas2016new, chan2014optimal}. Although this reduction
introduces a few complications to deal with, our main challenge is
actually testing under the $\ell_2$ distance.

{\em Closeness testing} ($\CT2p$) is technically the simpler problem,
but it already illustrates some phenomena, in particular, how to
leverage a larger number of samples to improve communication. To
estimate the $\ell_2$ distance between the 2 unknown distributions, we
compute the $\ell_2$ distance approximation between the given {\em
  samples} of these distributions.  In order to approximate the latter
in the 2-party setting, we use the {\em $\ell_2$ sketching} tools \cite{AMS}.
The crux is to show that we can tolerate a cruder $\ell_2$
approximation if we are given a larger sample size.  Since the
complexity of $(1+\alpha)$-approximating the $\ell_2$ distance is
$\Theta(1/\alpha^2)$, we obtain an improvement in communication that
is quadratic in the number of samples.

{\em Independence Testing} ($\IT2p$) is more challenging since any
distance approximation would need to be established based on the
distribution(s) implicitly defined via the {\em joint samples}, split
between Alice and Bob, and hence our approximation techniques above is
not sufficient. Instead, we develop a reduction from a large, $[n]
\times [m]$, alphabet problem, to a smaller alphabet problem, which
can be efficiently solved by communicating fewer samples. This is
accomplished by sampling a rectangle of the joint alphabet, and
showing that such a process, when combined with the {\em split-set}
technique from \cite{diakonikolas2016new}, generates {\em
  sub-distributions} (defined later) which satisfy some nice
properties. We then show one can test the original distribution
$p=(a,b)$ over a ``large'' domain of size $[n] \times [m]$ for
independence by distinguishing closeness of 2 simulated distributions
$\hat{p}, \hat{q}$, defined on a smaller domain of size $[l] \times
[m]$, where $l = \tilde{\Theta}_\eps(n^3m/t^3+n^2m/t^2+1)$. We show it
is possible for Alice and Bob to simulate joint samples from $\hat{p}$
and $\hat{q}$ using $O(1)$ communication per sample, after they have
down-sampled letters from one of the marginals.

The trade-off on communication--vs--samples emerges from two
compounding effects: 1) balancing the size of the target rectangle
with the expected number of available samples over such rectangle; and
2) the additional advantage from a tighter bound on the $\ell_2$ norm
of $\hat{q}$. Each of the above independently generates linear
improvement in communication with more samples. The latter advantage,
however, is helpful only while $t = O(n)$, and therefore we benefit
from quadratic improvement in that regime, and linear improvement
thereafter.

\smallskip
\paragraph{Lower bounds on communication.} We note that the lower bounds on
communication of testing problems present a particular technical
challenge: for testing problems, the inputs are i.i.d.~samples from
some distributions. This is more akin to the average-case complexity
setup, as opposed to ``worst case'' complexity as is standard for
communication complexity lower bounds.

We manage to prove such testing lower bounds for the {\em Closeness
  Testing} problem ($\CT2p$). While our lower bound is, at its core, a
reduction from some ``hard 2-party communication problem'', our main
contribution is dealing with the above challenge.  One may observe
that a ``hard 2-party communication problem'' is hard under a certain
input distribution (by Yao's minimax theorem), and hence a reduction
algorithm would also produce a hard distribution on the inputs to our
problem. However, apriori, it is hard to ensure that the resulting
input distribution resembles anything like a set of samples from
distributions $a,b$. For example, the inputs may have statistical
quirks that actually depend on whether it is a ``close'' or
``$\eps$-far'' instance, which a reduction is not able to generate
without knowing the output. Indeed, this is the major technical
challenge to overcome in our reduction.

At a high level, the role of the ``hard problem'' is played by a
variant of the well-known two-way Gap Hamming Distance (GHD) problem
\cite{chakrabarti2012optimal, DBLP:journals/toc/Sherstov12}. The known
GHD lower-bound variants are insufficient for us precisely because of
the above challenge---we need a better control over the actual hard
distribution, and in particular, for the $a=b$ instance, we need the
GHD input vectors to be at some fixed distance. Therefore, we study
the following {\em Exact GHD} variant: given $x,y\in\{0,1\}^n$, with
$\|x\|_1=\|y\|_1=n/2$, distinguish between $\|x-y\|_1=n/2$ versus
$\|x-y\|_1\in [n/2+\beta, n/2+2\beta]$. We show there exists some
$\beta\in [\Omega(\sqrt{n}), O(\sqrt{n\log n})]$ for which
communication complexity must be $\tilde \Omega(n)$, by adapting the
proof of \cite{DBLP:journals/toc/Sherstov12}. We note that such a
lower bound is far from apparent even in the 1-way communication model
(in contrast to the standard GHD, whose 1-way lower bound follows
immediately from the Indexing problem).

Using one instance of {\em Exact GHD}, our
reduction performs a careful embedding of this hard instance into the
samples from distributions $a,b$, while patching the set of samples to
look like i.i.d.~samples from the two distributions. While we don't
manage to get the output of the reduction to look precisely like
i.i.d.~samples from $a,b$, our reduction produces two sets of size
$\Poi(t)$ whose 
distribution is within a small statistical distance from the
distribution of two set of samples that would be drawn from two
distributions $a$ and $b$ which are either \say{equal} (when
$\|x-y\|_1=n/2$) or \say{far} (when $\|x-y\|_1\in [n/2+\beta,
  n/2+2\beta]$). 

We note that an immediate corollary of our lower bound is an
alternative proof that $\tilde\Omega(n^{2/3})$ samples are necessary
to solve closeness testing (in the vanilla testing setting), albeit
the resulting bound is off by polylogarithmic factors from the best
known ones \cite{chan2014optimal,diakonikolas2016new}.

For {\em Independence testing} ($\IT2p$), we focus on the lower bound
for unbounded number of samples, which is a bit simpler to deal with
as it becomes a ``worst-case problem'' (worst-case input joint
distributions $(a,b)$). We are able to show such a hardness result
under one-way communication only. Our $\Omega(\sqrt{m})$ lower bound
uses the Boolean Hidden Hypermatching (BHH) problem
\cite{doi:10.1137/1.9781611973082.2}. From a single BHH hard instance,
we generate an unbounded number of sample pairs $(A_i,B_i)$, where
the samples are drawn from some joint distribution $p =
(a,b)$. Depending on the BHH instance, 
$p$ is either exactly product or far from product distribution.

We conjecture that our entire trade-off for the Independence problem
is tight. The proof of this conjecture would have to overcome the above
challenge of lower bounds for statistical inputs, for finite $t$.

\medskip
\paragraph{Securing the communication protocols.}
Once low-communication insecure protocols have been designed, one may
try to convert the protocols to secure ones using generic
cryptographic techniques. The latter includes various techniques for secure
two-party computation (\cite{Yao82b} and followup work),
fully homomorphic encryption (\cite{Gentry09} and followup work), or
homomorphic secret sharing (\cite{BGI16,BGI17}). 
However, a na\"ive  application of such 
techniques
will  blow up the communication bounds to be at least linear in the input
size (which is prohibitive in our context), as well as possibly
requiring strong assumptions, a high computation complexity, or not
being applicable to arbitrary computations. The constraint of
low-overhead, combined with other considerations, requires design of
custom secure protocols.

Our starting point is a technique that falls into the latter category:
secure circuits with ROM~\cite{DBLP:journals/corr/cs-CR-0109011}, a
technique that can transform an insecure 2-party protocol to a secure
one with a minimal blow-up in communication, and uses a weak
assumption only (OT).  In order to obtain an efficient protocol,
however, it only applies to computations expressible via a very small
circuit, whose size is proportional to the target communication, %
with access to a larger read-only-memory (ROM) table.  Thus, the main challenge
becomes to design two-party testing protocols that fit this required
format.

For {\em Closeness Testing} ($\CT2p$), we begin with our
low-communication non-secure protocol, and adapt it to be secure by
designing a small circuit. One of the main difficulties in designing
such a circuit is that, in the $\ell_1$- to $\ell_2$-testing
reduction, Alice and Bob need to agree on an alphabet, which depends
on their inputs, without compromising the inputs themselves. To bypass
this, and other issues, we allow Alice and Bob to perform some
off-line work and prepare some polynomial-size inputs (in ROM). First,
we devise a method for Alice and Bob to generate a combined split set
$S$ (discussed later) by having each of Alice and Bob contribute
sampled letters to $S$. Second, we securely approximate the $\ell_2$
distance of Alice and Bob's original, un-splitted samples using
techniques from \cite{Indyk:2006:PPA:2180286.2180304}. Finally, we
adjust our approximation by accounting for some small number of
letters which differ from the original alphabet or which cannot be
approximated efficiently. To allow for such adjustment to be easily
accessible by a bounded-size circuit, Alice and Bob prepare inputs for
all possible scenarios. The main focus of our analysis goes into
proving our construction adds only poly-logarithmic factors in
communication over the non-secure protocol. Our adapted secure
algorithm turns out to deviate significantly from the simpler,
non-secure $\CT2p$ algorithm; hence we present the two protocols
separately.

For {\em Independence Testing} ($\IT2p$), the main challenge is in
designing the communication-efficient (insecure) protocol, while
adapting it to a secure one is somewhat simpler. One challenge
is to securely accomplish the aforementioned alphabet reduction (as the randomness, together with the final output, might compromise Bob's input). We
address this issue by having the entire process of
sampling a sub-distribution and pairing of samples to be
done over a secure circuit. The sampled alphabet, however, might be too large to communicate, and to overcome this obstacle, we devise a method for sampling from the non-empty
letters of Alice in a way that is distributed as if we would be
sampling directly over the entire alphabet. For this problem, the insecure and
secure protocol are more similar, and therefore we directly describe the secure
version.  %
\noindent

\section{Preliminaries}
\paragraph{Notation.} Throughout this paper we denote distributions in small letters,
and distribution samples in capital letters. Unless stated otherwise,
any distribution is on alphabet $[n]$, and domain elements of $[n]$
are addressed as letters. In addition, whenever discussing secure
protocols, we denote distributions as $a, b$, and for generic settings
as $p,q$. Similarly, secure vectors of distribution samples are
denoted as $A,B$ while generic ones are denoted as $X,Y$.

We also denote any multiplicative error arising from
approximation as $1+\alpha$, and any error or distance of/between
distributions as $\epsilon$.  Unless stated otherwise, distance and
norms are referring to the Euclidean distance and $\ell_2$ norms.

\paragraph{Split Distributions.}
We use the concept of split distributions from
\cite{diakonikolas2016new} to essentially reduce testing under
$\ell_1$ distance to testing under $\ell_2$ distance.
\begin{definition}
\label{def:splitDist}
Given a probability distribution $p$ on $[n]$ and a multiset $S$ of
items from $[n]$, define the {\em split distribution} $p_S$ on $[n +
  |S|]$ as follows. For $i\in[n]$, let $a_i$ be equal to 1 plus the
number of occurrences of $i$ in $S$; note that $\sum_{i=1}^{n} a_i = n
+ |S|$. We associate the elements of $[n + |S|]$ to elements of the
set $E$ = $\{(i, j) : i \in [n], 1 \leq j \leq a_i\}$. Now the
distribution $p_S$ has support $E$ and a random draw $(i, j)$ from
$p_S$ is sampled by picking $i$ randomly from $p$ and $j$ uniformly at
random from $[a_i]$.
\end{definition}
Recall from \cite{diakonikolas2016new} that split distributions are used to upper bound the $\ell_2$ norm of an underlying distribution while maintaining its $\ell_1$ distance to other distributions:
\begin{fact}[\label{split_fact}\cite{diakonikolas2016new}]
Let p and q be probability distributions on $[n]$, and $S$ a given multiset of $[n]$. Then
\begin{itemize}

\item We can simulate a sample from $p_S$ or $q_S$ by taking a single sample from p or q, respectively.
\item $\|p_S - q_S\|_1 = \|p - q\|_1.$
\end{itemize}
\end{fact}
\begin{lemma}[\cite{diakonikolas2016new}]\label{split_lemma}
Let $p$ be a distribution on $[n]$. Then: (i) For any multisets $S \subseteq S'$ of $[n], \normt{p_{S'}}\leq \normt{p_{S}}$, and (ii) If $S$ is obtained by taking $Poi(m)$ samples from $p$, then $\E[\normt{p_S}^2] \leq 1/m$.
\end{lemma}

\medskip
We provide further preliminaries needed for the secure versions of our
protocols in Section~\ref{subs_crypt_tools}.

\section{Closeness Testing: Communication-Efficient Protocol}
\label{sec_priv}

In this section we consider the closeness testing problem $\CT2p$,
focusing first on the 2-party communication complexity only. In
Section~\ref{sec:secureCT}, we modify the protocol to make it secure.

As mentioned in the introduction, one way to obtain a protocol is to
use unequal-size closeness testing, where Alice has $s$ samples and
Bob has $t$ samples: Alice just sends her $s$ samples to Bob, and Bob
invokes a standard algorithm for closeness testing.
Using the optimal bounds from, say,
\cite{diakonikolas2016new}, we get the following
trade-off for fixed $\eps$: $s=\tilde O(n/\sqrt{t})$, with the condition
that $s,t\ge \sqrt{n}$.

Here we obtain a {\em polynomially smaller} communication complexity,
$s=\tilde O_\eps(n^2/t^2)$, whenever $t$ is above the
information-theoretic minimum on the number of samples. In
Section~\ref{sec:ctLB}, we show a nearly-matching lower bound.

\subsection{Tool: approximation via occurrence vectors}

Our protocol uses the framework introduced in
\cite{diakonikolas2016new}, allowing us to focus on the $\ell_2$
testing problem. For $\ell_2$ testing, we show that, for two discrete
distributions $p, q$, we can approximate their $\ell_2$ distance by
approximating the $\ell_2$ distance of their respective sample
occurrence vectors, defined as follows.

\begin{definition}
Given $t$ samples of some distribution $p$ over $[n]$, we define {\em
  the occurrence vector} $X \in [t]^n$ such that $X_i$ represent the
number of occurrences of element $i\in[n]$ in the sample set.
\end{definition}

The following lemma bounds how well we need to estimate the $\ell_2$
distance between occurrence vectors in order to distinguish between
$p=q$ vs.~$\|p-q\|_1\ge \eps$. Overall it shows that the more samples
we have, the less accurate the $\ell_2$ estimation needs to be.  Using
the framework from \cite{diakonikolas2016new}, for now it is enough to
assume that the $\ell_2$ norm of both $p$ and $q$ is bounded by $U <
1$.

\begin{lemma}\label{app_dist_test}
	Let $p, q$ be distributions over $[n]$ with
        $\normt{p},\normt{q}\le U$ for some $U<1$. There exists a
        constant $C>0$, such that for $t=O(U \cdot n
        \cdot \epsilon^{-2})$, and $\alpha
        = \Omega(U)$, given $\Delta$ which is $(1\pm\alpha)$-factor
        approximation of $\normt{X-Y}^2$ where $X, Y$
        represent the occurrence vectors of $t$ samples drawn from $p,
        q$ respectively, then, using $\Delta$, it is possible to distinguish
        whether $p=q$ versus $\normo{p-q}>\epsilon$ with 0.8 probability.
\end{lemma}

The actual distinguishing algorithm is simple: for fixed $\alpha =
\Omega(U)$, we merely compare $\Delta$ to some fixed threshold $\tau$
(fixed in the proof below). The intuition is that for a given number
of samples, we have some gap between the range of possible distances
$\normt{X-Y}^2$ for each of the cases. If the number of samples is
close to the information-theoretic minimum
\cite{chan2014optimal}, then the gap is minimal and we
need to calculate almost exactly the distance, hence approximating the
distance between occurrence vectors doesn't help. However, as the
number $t$ of samples increases, so does gap between the ranges, and
we can afford a looser distance approximation.

\begin{proof}[Proof of Lemma \ref{app_dist_test}]
  Given $t=O(U/\epsilon'^2)$ samples from each $p, q$, according to
  \cite[Proposition 3.1]{chan2014optimal}, the estimator
  $\frac{\sqrt{\sum_i(X_i-Y_i)^2 - X_i - Y_i}}{t}$ is a
  $\max\{\epsilon',\normt{p-q}/8\}$ additive approximation of
  $\normt{p-q}$ with 0.9 probability.
	Setting $\epsilon' = \epsilon/8\sqrt{n}$, we get that:
\begin{align*}
	\normt{p-q} = 0 &\Rightarrow \normt{X-Y}^2 \leq \frac{\epsilon^2 t^2}{4n} + 2t \\ 
	\normt{p-q} > \epsilon/\sqrt{n} &\Rightarrow \normt{X-Y}^2 \geq \frac{3\epsilon^2 t^2}{4n} + 2t 
\end{align*}
Suppose $\Delta$ is such that $\tfrac{\Delta}{\normt{X-Y}^2}\in (1-\alpha,
1+\alpha)$. If $\normo{p-q}= 0$, then $\normt{p-q} = 0$ and hence
\begin{align*}
	 \Delta &\leq (1+\alpha)(\frac{\epsilon^2t^2}{4n} + 2t) 
	 \leq t(\frac{\epsilon^2t}{4n} + 2 + 2\alpha +
         \frac{\alpha\epsilon^2t}{4n});
\end{align*}
and if $\normo{p-q} > \epsilon$, then $\normt{p-q} >
\epsilon/\sqrt{n}$ and hence
\begin{align*}
	\Delta &\geq (1-\alpha)(\frac{3\epsilon^2t^2}{4n} + 2t) 
	\geq t(\frac{3\epsilon^2t}{4n} + 2 - 2\alpha - \frac{3\alpha\epsilon^2t}{4n}).
\end{align*}
We distinguish the two cases, by comparing $\Delta$ to
$\tau=\frac{\epsilon^2t^2}{2n} + 2t$: namely $p=q$ iff $\Delta\le \tau$.
Hence we just need to ensure that
$$
t(\frac{3\epsilon^2t}{4n} + 2 - 2\alpha -
\frac{3\alpha\epsilon^2t}{4n}) -\tau \ge \tau - t(\frac{\epsilon^2t}{4n} + 2 +
2\alpha + \frac{\alpha\epsilon^2t}{4n})$$
Hence we need that $\frac{\epsilon^2t}{4n} -
2\alpha - \frac{3\alpha\epsilon^2t}{4n}\ge 0$, or $\alpha \le \tfrac{\eps^2t}{4n\cdot
  (2+3\eps^2t/4n)}$. Since $t=O(Un\eps^{-2})$, the conclusion follows.
\end{proof}

\subsection{Communication vs number of samples}

We now provide a (non-secure) protocol for $\CT2p$ with a trade-off between
communication and number of samples.
\begin{theorem}[Closeness, insecure]\label{thm_2pct}
Fix $n>1$ and $\eps \leq 2$. There exists some constant $C>0$ such that for all $t \geq C \cdot \max\left(n^{2/3}\cdot \epsilon^{-4/3}, \sqrt{n}\cdot
\epsilon^{-2}\right)$, the problem $\CT2p_{n,t,\epsilon}$ can be solved using
$\tilde O\left(\frac{n^2}{t^2\epsilon^4} + 1\right)$ bits of
communication.
\end{theorem}
The protocol uses Lemma~\ref{app_dist_test} as the main algorithmic
tool and proceeds as follows.  Bob generates multi-set $S$ using
samples from $b$ and sends $S$ to Alice. Then, Alice and Bob each
simulates samples from $a_S$ and $b_S$ respectively, and they together
approximate the $\ell_2$ difference of the resulting occurrence
vectors using sketching methods \cite{AMS}.\\

\fbox{%
\parbox{450pt}{
\textbf{Non-Secure $\CT2p(a,b,t)$} \\
Alice's input: $t$ samples from $a$ \\
Bob's input: $t$ samples from $b$
\begin{enumerate}
	\item Fix $\alpha= \Omega(t \cdot \epsilon^2/n)$.
	\item Bob generates multi-set $S$ using  $Poi(\tfrac{n^2}{t^2\epsilon^4})$  samples from $b$.
	\item Bob sends $S$ to Alice.
	\item Alice and Bob recast their samples as being from
          distributions $a_S, b_S$ (see Def.~\ref{def:splitDist}), and
          set $A_S, B_S$ to be the respective occurrence vectors.
        \item \label{it:L2normComparison} Alice and Bob each estimate $\|a_S\|_2$ and $\|b_S\|_2$
          up to factor 2; if the two estimates are not within factor $4$, output \say{$\eps$-FAR};
	\item Alice and Bob approximate $\Delta=\normt{A_S-B_S}^2$ up
          to $(1+\alpha)$ factor, using, say, \cite{AMS}.
	\item If $\Delta$ is less than
          $\tau=\tfrac{\eps^2t^2}{2n}+2t$ output \say{SAME}, and,
          otherwise, output \say{$\eps$-FAR}.
\end{enumerate}
}}

\begin{proof}[Proof of Theorem \ref{thm_2pct}]
We note that, according to Lemma \ref{split_lemma},
$\E[\normt{b_S}^2] = t^2\epsilon^4 / n^2$ and hence $\normt{b_S}^2 =
O(t^2\epsilon^4 / n^2)$ with at least 90\% probability. Furthermore, since $t = \Omega(\sqrt{n}/\epsilon^2)$,
we have that $|S|= O(n)$ with high probability. From now on, we condition on
these two events. 

If $\normt{a_S} \neq \Theta(\normt{b_S})$ then distributions are
different and we output \say{$\eps$-far} is step
\ref{it:L2normComparison}. Otherwise, we have that
$\|a_S\|_2^2=O(\|b_S\|_2^2)=O(t^2\epsilon^4 / n^2)$. Hence we can use
Lemma \ref{app_dist_test}, where $U=O(t\epsilon^2 / n)$ and
$\alpha=\Omega(U)$, to claim the correctness of the protocol.

We now analyze the communication used by the protocol:
	\begin{enumerate}
		\item communicating $S$ takes $|S|\log n =
                  \tilde O(n^2/t^2\epsilon^4)$ bits with high probability.
		\item estimating $\Delta$ up to approximation
                  $1+\alpha$ takes $\tilde O(1/\alpha^2) = \tilde
                  O(n^2/t^2\epsilon^4)$ bits, using standard $\ell_2$
                  estimation algorithms \cite{AMS, KOR00}.
	\end{enumerate}
\end{proof}

\begin{remark}
Another application of this protocol is that it can be
  simulated by a {\em single party\/} to obtain space-bound {\em
    steaming algorithm\/} with the same space/sample trade-offs. While
  we are not formalizing this argument in this paper, this can
  essentially be done by storing $S$ and sketching $\normt{A_S -
    B_S}^2$.
\end{remark}

\section{Closeness Testing: Communication Lower Bounds}
\label{sec:ctLB}

We now prove that the protocol for $\CT2p$ from Section~\ref{sec_priv}
is near-tight, showing the following theorem:

\begin{theorem}
\label{thm:closeDisj}
Let $a, b$ be some distributions over alphabet $[n]$, where Alice and
Bob each receive $\Poi(t)$ samples from $a, b$ respectively, for $t\le
n/\log^c n$ for some large enough $c>1$. Then any (two-way)
communication protocol $\Pi$ that distinguishes between $a=b$ and
$\normo{a-b}\geq 1/2$ requires $s=\tilde\Omega(n^2/t^2)$ communication.
\end{theorem}

Intuitively, our proof formalizes the concept that in testing
distributions for closeness, ``collisions is all that matters'', {\em
  even in the communication model}. This is similar to the intuition
from the ``canonical tester'' from \cite{valiant2011testing}, which
shows a similar principle when all the samples are accessible. Our
result can be seen to extending it to saying that the canonical tester
is still the best even if we have more-than-strictly-necessary number
of samples that we could potentially compress in a communication
protocol.

To prove the theorem, we rely on the
following communication complexity lower bound, which is a variant of
the GapHamming lower bound \cite{chakrabarti2012optimal,
  DBLP:journals/toc/Sherstov12}. A somewhat surprising aspect of this
GapHamming variant is that, unlike for the standard GapHamming, we are
not aware of a lower bound for one-way communication that would be
simpler than the two-way proof from the lemma below.

\newcommand{\cA}{{\cal A}}

\begin{lemma}
\label{lem:gapHamming}
Let $n\ge 1$ be even. There exists some
$\beta=\beta(n)\in[\Theta(\sqrt{n}), \Theta(\sqrt{n\log n})]$, %
satisfying the following.  Consider a two-way
communication protocol $\cA$ that, with probability at least $0.9$,
for $x,y\in\{0,1\}^n$ with $\|x\|_1=\|y\|_1=n/2$, can distinguish
between the case when $\|x-y\|_1=n/2$ versus $\|x-y\|_1-n/2\in
[\beta,2\beta]$. Then $\cA$ must exchange at least $\Omega(\tfrac{n}{\log
n\cdot \log\log n\cdot \log\log\log n})$ bits of communication.
\end{lemma}

The proof of this lemma is presented in Appendix
\ref{sec:gapHamming}. 

Now the idea is to reduce an instance of the GapHamming input from
Lemma~\ref{lem:gapHamming} to an instance of closeness testing by
carefully molding the input $(x,y)$ into a couple of related
occurrence vectors $(A,B)$.%
We use the following estimate on the statistical distance between
Multinomial and Poisson random variables. 
\begin{definition}
Consider $n,k\ge 1$, as well as a vector $\vec{p}\in \R^k_+$, where
$\sum_{i=1}^k p_i\le 1$. Then let $(M_1,\ldots M_k)=\Mult(n;\vec{p})$
be the $k$-dimensional random variable obtained by drawing a
Multinomial random variable with parameters $n$ and probability vector
$(1-\sum_{i=1}^k p_i, \vec{p})$, and dropping the first coordinate.
\end{definition}

\begin{theorem}[\cite{barbour1988stein}]
\label{thm:barbour88}
Let $n,k\ge 1$, as well as a vector $\vec{p}\in \R^k_+$, where
$p=\sum_{i=1}^k p_i\le 1$. Consider the random variable $(M_1,\ldots
M_k)$ to drawn from the Multinomial $\Mult(n; \vec{p})$. Also consider
the Poisson random variable $P=(P_1,\ldots, P_k)$ where $P_i\sim
\Poi(np_i)$. Then the variables $M=(M_1,\ldots M_k)$ and $(P_1,\ldots
P_k)$ are at a statistical distance of $O(p\log n)$.
\end{theorem}

\begin{proof}[Proof of Theorem \ref{thm:closeDisj}]
Consider some input vectors $x,y$,
of length $m=\tfrac{n^2}{t^2\log^3n}$, to the GapHamming problem from
above. Let
$\Delta=\beta(m)=\Omega(\sqrt{m})$, and
$\delta=\tfrac{1}{2}(\|x-y\|_1-m/2)\in \{0\}\cup
[\Delta/2,\Delta]$.
The case of $\delta=0$ will correspond to ``same'' case (i.e. $a=b$),
and $\delta \in [\Delta/2,\Delta]$ --- to ``far'' case (i.e. $\normo{a-b} \in [1/2,1]$).

Fix $d=n/10$ and $l= C\cdot t \cdot \log n$ (where $C$ is some constant that we shall fix later), which have the following meaning: each
distribution $a,b$ has half mass over $[d]$ items uniformly (called
{\em dense items}), and the other half on $[l]$ items uniformly (called
{\em large items}). When $a=b$, these are the same items, and when $a\neq
b$, the large items are the same while the dense items have supports with a large difference. In
particular, the dense items are supported on sets $S_A, S_B$
respectively, with $|S_A|=|S_B|=d$, and $S_A\cap S_B=d\cdot
\tfrac{\Delta-\delta}{\Delta}$; we hence also have that $|S_A\setminus
S_B|=d\cdot \tfrac{\delta}{\Delta}$.  

Now for $i\ge 0$, let $D(i) = \Pr[\Poi(t/2d)=i]$, i.e.,
probability that a dense number is sampled $i$ times. For simplicity,
we write $D(i,j)=D(i)\cdot D(j)$. Similarly we define
$L(i)=\Pr[\Poi(t/2l)=i]$ and $L(i,j)=L(i)\cdot L(j)$. We also set
$k=\Theta(\log n)$, which should be thought of as an upper bound on the count
of any fixed item (with high probability).

The algorithm constructs the occurrence vectors $A,B$ iteratively as
follows. We note that all random variables are chosen using shared randomness. Let $m_c=m/4-\Delta$.
\begin{enumerate}
\item
For each $i,j\in\{1,\ldots k\}$, and for each letter $c\in [n]$, we
generate $\Poi(\tfrac{d}{\Delta}\cdot D(i,j))$ copies of
letter $c$: Alice replaces $1$ with $i$ and Bob replaces $1$ with
$j$ (both leaving 0s intact);
\item
For each $i\in \{1,\ldots k\}$, generate $\Poi(d\cdot D(i,0))$ pairs
$(i,0)$, and similarly-distributed number of pairs $(0,i)$;
\item
For each $i,j\in \{1,\ldots k\}$, generate $\Poi(l\cdot
L(i,j)-m_c\cdot \tfrac{d}{\Delta}D(i,j))$ pairs $(i,j)$;
\item
For each $i\in \{1,\ldots k\}$, generate $\Poi(l\cdot
L(i,0)-\tfrac{m}{4}\cdot \tfrac{d}{\Delta}\sum_{j=1}^k D(i,j))$ pairs $(i,0)$, and
similarly-distributed number of pairs $(0,i)$.
\item
Fill in the required number of $(0,0)$ pairs so that $A,B$ have length
precisely $n$;
\item
Randomly permute the letters of $A,B$ (using shared randomness).
\end{enumerate}

\begin{claim}
All the Poisson random variables from above are properly defined---in
particular, they have positive argument.
\end{claim}
\begin{proof}
We only need to prove this for steps 3 and 4 as the other ones are
obvious. Indeed, for $i,j\ge 1$:
$$
l\cdot L(i,j)=t\log n \cdot (\Omega(1/\log n))^{i+j}=t/\log n\cdot(\Omega(1/\log n))^{i+j-2},
$$
whereas,
$$
m_c\cdot \tfrac{d}{\Delta}D(i,j)=O(\sqrt{m}\cdot
n\cdot (t/2d)^{i+j})\le \tfrac{n^2}{t\log^{1.5}n}\cdot O(t^2/n^2)\cdot
(O(t/n))^{i+j-2}\le \tfrac{t}{\log^{1.5}n}(O(t/n))^{i+j-2}.
$$
Thus $l\cdot L(i,j)-m_c\cdot \tfrac{d}{\Delta}D(i,j)\ge 0$ for all
$i,j\ge1$.

Similarly, for step 5, for $i\ge 1$, we have:
$$
l\cdot L(i,0)=\Omega(t\cdot (O(1/\log n))^{i-1}),
$$
whereas,
$$
m/4\cdot \tfrac{d}{\Delta}\sum_{j\ge 1}D(i,j)
\le
O(\sqrt{m}\cdot n\cdot \sum_{j\ge 1} (O(t/n))^{i+j})
\le
O(\tfrac{n^2}{t\log^{1.5}n}\cdot (O(t/n))^{i+1})
\le
O(\tfrac{t}{\log^{1.5}n}\cdot (O(t/n))^{i-1}).
$$
We again have $l\cdot L(i,0)-m/4\cdot \tfrac{d}{\Delta}\sum_{j\ge
  1}D(i,j)\ge 0$ as required.
\end{proof}

We now prove the core of the reduction: that the distribution of
$(A,B)$ is close to occurrence vectors of $\Poi(t)$ i.i.d. samples from
$(a,b)$, such that $a=b$ if $\|x-y\|_1=m/2$, and similarly,
$\normo{a-b} \geq 1/2$ when $\|x-y\|_1\ge m/2+\beta$.  We will prove
that, for distribution of (co-)occurrences of large items is nearly
same in the two instances; and similarly for the dense items.

We partition the coordinates of $(x,y)$ in the following four groups, each
corresponding to either occurrences of dense or large items:
\begin{itemize}
\item
large: $m_c=m/4-\Delta$ coordinates for each of $(1,1)$ and $(0,0)$
coordinate pairs (i.e., coordinates $i\in[m]$ where $(x_i,y_i)=(1,1)$ or
$(x_i,y_i)=(0,0)$);
\item
large: $m/4$ coordinates for each of $(1,0)$ and $(0,1)$
pairs;
\item
dense: $\Delta-\delta$ coordinates for each of $(1,1)$ and $(0,0)$
pairs;
\item
dense: $\delta$ coordinates for each of $(1,0)$ and $(0,1)$
pairs. 
\end{itemize}
Note that this accounts for all coordinates for a pair $x,y$ such that
$\|x-y\|_1=m/2+\delta$.  \\

We now analyze the distribution of occurrences/collisions for each of
large and dense items in the generated vectors $(A,B)$, and show that,
for each of large/dense items, the distribution is same as if these
are occurrences of items coming from distributions $a,b$ defined
above. In particular, for, say, large items, we consider the
distribution of counts $c_{i,j}$, where $i+j>0$, where $c_{i,j}$ is
the number of large items which where sampled $i$ times on Alice's
side and $j$ times on the Bob's side; we will refer to them as $(i,j)$
occurrence pairs. We then show that, the distribution of
$(c_{i,j})_{i+j>0}$ in $(A,B)$ is Poisson-distributed, whereas, if it
were the occurrence pairs vector of samples drawn from $a,b$, then the
distribution is a Multinomial. We then use Theorem~\ref{thm:barbour88}
to conclude that the two distributions are statistically close. Note
that the identify of items is not important, as the items are randomly
permuted inside the domain, for both $A,B$ as well as in distributions $a,b$.

{\bf Large items in $(A,B)$.}  We analyze the large items first. We use the fact
that sum of Poisson distributions is again Poisson. For any
$i,j\in\{1,\ldots k\}$, the number of large $(i,j)$ pairs is
distributed as: $\Poi(m_c\cdot \tfrac{d}{\Delta}D(i,j))$ (from the
first step: there are $m_c$ coordinate pairs $(1,1)$), plus
$\Poi(l\cdot L(i,j)-m_c\cdot \tfrac{d}{\Delta}\cdot D(i,j))$ (from the
third step).  Thus the number of large $(i,j)$ pairs is distributed as
$\Poi(l\cdot L(i,j))$.

Similarly, say, considering occurrence pair $(i,0)$ (a symmetric argument applies for $(0,i)$), the number of
large $(i,0)$ pairs is distributed as $\Poi(m/4\cdot
\tfrac{d}{\Delta}\sum_{j=1}^k D(i,j))$ (from the first step: there are
$m/4$ coordinate pairs $(1,0)$), plus $\Poi(l\cdot L(i,0)-m/4\cdot
\tfrac{d}{\Delta}\sum_{j=1}^k D(i,j))$ (from step 4). This again
amounts to $\Poi(l\cdot L(i,0))$. 

{\bf Large items in $(a,b)$.}
Let us now contrast these counts to the one would get from the ``real
counts'' of the large items in the distribution $(a,b)$ defined as
above. The latter is a Multinomial $M^L=\Mult(l;\vec{p}_L)$ where
$\vec{p}_L=(L(i,j))_{i,j\ge 0; i+j>0}=(L(1,0), L(0,1), L(1,1), L(2,0), L(2,1), ...)$. We now can
use Theorem~\ref{thm:barbour88}, to conclude that the the TV-distance
between $M^L$ and the distribution $\Poi(l\cdot \vec{p}_L)$ is bounded
by: $O(\log n)\cdot \sum_{i,j\ge 0; i+j>0} L(i,j)\le O(\log n)\cdot \sum_{i,j\ge 0;
  i+j>0}(t/2l)^{i+j}\le O(1)/C \leq 0.01$ (by choosing $C$ to be a large enough constant). 

{\bf Dense items in $(A,B)$.}
Let's analyze the distribution of dense items now. For $i,j\ge 1$, the
distribution of the number of $(i,j)$ dense occurrence pairs is
$\Poi((\Delta-\delta)\cdot \tfrac{d}{\Delta}D(i,j))$ as there are
$\Delta-\delta$ coordinate pairs $(1,1)$.  Now consider the case of
$(i,0)$ occurrence pairs of dense items, for $i\in[k]$. Their count is
distributed as: $\Poi(\delta\cdot
\tfrac{d}{\Delta}\sum_{j=1}^kD(i,j))$ (from the first step: there
are $\delta$ coordinate pairs $(1,0)$), plus $\Poi(d\cdot D(i,0))$
(from the second step). This amounts to:
$$
\Poi(d\cdot \tfrac{\delta}{\Delta}\sum_{j=1}^k D(i,j)+d\cdot D(i,0)).
$$

{\bf Dense items in $(a,b)$.}
Again, let's compare these counts to the ``real counts'' that would
occur for the distributions $(a,b)$. The latter distribution can be
thought of as three distributions: corresponding to items in $S_A\cap
S_B$, to items in $S_A\setminus S_B$, and items in $S_B\setminus S_A$.
The occurrence counts for items in $S_A\cap S_B$ are distributed as a
Multinomial $M^{D,i}$ with parameters $|S_A\cap S_B|=d\cdot
\tfrac{\Delta-\delta}{\Delta}$  and probability vector
$\vec{p}_D=(D(i,j))_{i,j\ge 0; i+j>0}$. By
Theorem~\ref{thm:barbour88}, the TV-distance between
$M^{D,i}$ and the distribution $\Poi(d\cdot
\tfrac{\Delta-\delta}{\Delta}\cdot \vec{p}_D)$  is bounded by:
$O(\log n)\cdot \sum_{i,j\ge 0; i+j>0} D(i,j)\le O(\log n)\cdot \sum_{i,j\ge 0;
  i+j>0}(t/2d)^{i+j}\le O(1/\log n)$.

Note that the counts for $i,j\ge 1$ correspond to (1,1) pairs generating dense items in $A,B$
above. It remains to analyze the case of $j=0$ or $i=0$. Wlog,
consider $j=0$ and $i>0$ (ie, items that are only in Alice's
distribution). For the distributions $a,b$, the occurrence pairs
$(i,0)$ are distributed as a Multinomial $M^{D,a}=\Mult(|S_A\setminus
S_B|;\vec{p}_{DA})$, where $\vec{p}_{DA}=(D(i))_{i\ge 1}$. By
Theorem~\ref{thm:barbour88}, the TV-distance between $M^{D,a}$ and
$\Poi(d\cdot \tfrac{\delta}{\Delta}\cdot \vec{p}_{DA})$ is at most
$O(\log n)\cdot\sum_{i\ge 1} D(i)\le O(1/\log n)$.
Summing up with the above, and focusing on $(i,0)$ pairs, their
distribution in $(a,b)$ is at small distance to the distribution where
each pair $(i,0)$ is distributed as:
$$
\Poi(d\cdot \tfrac{\Delta-\delta}{\Delta}\cdot D(i,0)+d\cdot
\tfrac{\delta}{\Delta}\cdot D(i))=
\Poi(d\tfrac{\Delta-\delta}{\Delta}\cdot D(i,0)+d\tfrac{\delta}{\Delta}\cdot (D(i,0)+\sum_{j=1}^k D(i,j)))=
\Poi(d\cdot \tfrac{\delta}{\Delta}\cdot \sum_{j=1}^k D(i,j)+d\cdot D(i,0)).
$$

Thus we conclude that the distribution of occurrence pairs of the dense
items in vectors $A,B$ matches, up to a small total variation distance,
the distribution for $a,b$ as described above. 

This completes the proof of the lower bound.
\end{proof}

\section{Closeness Testing: Secure Communication-Efficient Protocols}
\label{sec:secureCT}

The protocol from Section~\ref{sec_priv} is clearly not secure (for example, Bob
sends Alice his multi-set $S$, which reveals information about his
input). In this section we show 
how to modify the protocol to make it secure, relying on standard
cryptographic assumptions. 
Before proceeding, we provide some preliminaries, followed by our 
general definition for secure computation of distribution testing.

\subsection{Cryptographic Tools and Preliminaries} \label{subs_crypt_tools}
We briefly review cryptographic tools and assumptions that we use. We
keep the discussion largely informal, and focus on the aspects most
relevant to our results.  We refer the reader to,
e.g.,~\cite{Gol01I,Gol04II} for more details and formal definitions of
standard primitives.

\medskip
\paragraph{PRG, OT, and Our Assumptions.}
A pseudorandom generator (PRG) is a deterministic function $G$ that
stretches its input length, such that if the input is selected
uniformly at random, the output is indistinguishable from uniform for
an appropriate class of distinguishers. 

In particular, the first cryptographic assumption that our secure
protocols rely on, is that 
there exists a PRG $G$
that can stretch $\polylog(m)$ bits to $m$ bits, and fools 
$\poly(m)$-sized circuits.
By default, this is what we mean when we refer to ``PRG'' in
the rest of the paper.

A 1-out-of-$m$ Oblivious Transfer (OT) protocol,
allows one party holding input $i\in[m]$  and another party holding
$s\in\{0,1\}^m$, to engage in a protocol where the first party obtains
as output the bit $s_i$, the other party obtains no output, and no
further information about $s$ or $i$ is leaked to the parties.

The second cryptographic assumption that our secure protocols rely on,
is that there exists a 1-out-of-$m$ OT protocol with communication
complexity $\polylog(m)$.
By default, this is what we mean when we refer to ``OT'' in the rest
of the paper.

We note that it is easy to extend a 1-out-of-$m$ OT protocol operating
on bits (as we defined above) to one that operates on words of size
$r$ (namely, with $s\in\{\{0,1\}^r\}^m$), 
paying a communication overhead that is linear in $r$.
Hence, under our OT assumption, there is such a protocol with
communication complexity that is polylogarithmic in $m$ and linear in
$r$.

PRG and OT are both standard cryptographic primitives, that can be
instantiated from various concrete number theoretic assumptions. The fact that
we require super polynomial stretch for the PRG means that we need to
assume subexponential hardness (and this is also the reason that we
don't use the OT assumption to generically obtain PRG, as we only need
OT with polynomial hardness). 

We note that if we wish to weaken the assumptions, we may assume 
polynomial-stretch PRG (from input of size $m^\delta$ to output of
size $m$ for a constant $\delta$) and sublinear-communication OT
(1-out-of-$m$ OT with communication complexity $m^\delta$).
Under these weaker assumptions, our secure protocols may incur
higher communication complexity, but would
still provide meaningful results (and remain secure).
Specifically, under a weaker PRG assumption, our protocols would
remain sub-linear always and would in fact keep the same complexity
for part of the trade-off range. Under a weaker OT assumption, our
protocols would remain sub-linear at least for some range of $\delta$
and $t$ (the number of samples).   

\medskip
\paragraph{Secure Computation.} 
Intuitively, secure computation allows two or more parties to evaluate some
function of their inputs, such that no additional information is
revealed to any party (or group of parties) beyond what follows from
their own inputs and outputs.

\smallskip\noindent
{\sc Defining Security: Simulation Paradigm. } 
Secure computation has been studied in many different settings.
The idea underlying the security definition in all these
settings, is trying to enforce that whatever the adversary can do in
the real-world, can also be achieved in an ideal world, where the
parties simply give their input to a trusted party, who hands them the
output.  Defining this formally is complex and various issues arise in
different settings. 

In this paper we focus on the simplest setting of a semi-honest (or ``honest-but-curious'') adversary, where the parties follow the
protocol faithfully, but try to use their transcripts to glean more 
information than intended.
Our  defintions and protocols
can be adapted to malicious adversary using standard techniques
\cite{GMW86,DBLP:journals/corr/cs-CR-0109011,Gol04II}.

In our setting, if semi-honest Alice and Bob want to compute a function
$f: \{0,1\}^* \times \{0,1\}^* \to \{0,1\}$, the definition of
security boils down to requiring the existence 
of an efficient simulator for each party, which can simulate the view
of the party (their input and transcript), from just the party's
input, random input, and output.

\smallskip\noindent{\sc Modular Composition. }
At times it is convenient to design a protocol in a modular way, where
the computation of a function $f$ may invoke a call to another
function $g$.  We will denote by $\Pi_f^g$ a protocol computing $f$,
with oracle gates to the function $g$. 
A composition theorem
\cite{Can00,Gol04II}  
proves that if the protocol $\Pi_f^g$ is a secure protocol for $f$,
and $\Pi_g$ is a secure protocol for $g$, then taking $\Pi^g_f$ and
replacing the oracle calls to $g$ with an execution of $\Pi_g$
results in a secure protocol $\Pi_f$  for $f$ (which doesn't make any oracle
calls). 

\smallskip\noindent{\sc Feasibility Results. }
Starting with \cite{Yao82b}, a large body of work has shown that any
function that can be computed, can be computed securely in various
settings.
In particular, two parties holding inputs $x$ and $y$ can securely
evaluate  any circuit $C(x,y)$ with communication $O_k(|C|)$
(where $k$ is the security parameter), under mild cryptographic
assumptions. %
Note that this communication complexity is at least linear (since the
circuit size is at least as large as the size of input and output),
while we will need sublinear communication protocols.
Other general techniques for secure function evaluation, where
communication does {\em not} depend on the circuit size, include
fully-homomorphic encryption (FHE) (\cite{Gentry09} and follow up work, see
\cite{B18} for a survey), and homomorphic secret sharing (HSS)
\cite{BGI16,BGI17}.  However, these transformations still require
communication that is linear in the length of input and output
(prohibitive in our context). They also 
require stronger assumptions, and (for FHE) require a high
computational overhead, or (for HSS) only apply for restricted classes
of circuits. 
Naor and Nissim \cite{DBLP:journals/corr/cs-CR-0109011}
showed a general way to transform insecure protocols to secure ones,
while preserving communication.  However, in general this
may introduce an exponential blowup in computation, which again will
not be sufficient for our needs.

While a naive application of these generic methods doesn't directly
work for us, we will make use of another result by Naor and Nissim,
adapting secure two party computation techniques to allow for
communication-efficient protocols whenever the computation can be
expressed as a 
(small) circuits with (large)  ROM.  

\smallskip\noindent
    {\sc Secure Circuit with ROM.}
    Consider the setting where each party has a
    table $R\in(\{0,1\}^r)^m$ (that is, $m$ entries of size $r$ each).
    Now consider a circuit $C$ that, in addition to usual gates, has
    lookup gates which allow to access any of the parties' ROM tables
    (on input $i\in[m]$ the gate will return the $r$-bit record at
    the requested party's $R(i)$).  
\begin{theorem}\label{thm_nn01}
\cite{DBLP:journals/corr/cs-CR-0109011}
If $C$ is a circuit with ROM, then it can be securely computed
with $\tilde{O}(|C|\cdot T(r, m))$ communication, where $T(r, m)$ is the communication of 1-out-of-$m$ OT on words of size $r$.
\end{theorem}
Thus, under our OT assumption,
a circuit with ROM can be securely evaluated with communication
complexity that is linear in $|C|$ and $|r|$, but polylogarithmic in
$m$. 
We will rely on this theorem in both of our secure distribution testing
constructions. Note that the main remaining challenge is to design
the protocol that can be expressed in a form where this theorem can be
applied.

\smallskip\noindent
{\sc Sampling an Orthonormal Matrix.}
We will use the following fact from
\cite{Indyk:2006:PPA:2180286.2180304}, proven in the context of
providing a secure approximation of the $\ell_2$ distance.

\begin{theorem}\cite{Indyk:2006:PPA:2180286.2180304}\label{randmat}
	Suppose we sample a random orthonormal $n\times n$ matrix $R$ (from a
        distribution defined by the Haar measure) but instead generate
        our randomness using a PRG $G$, rounding its entries to the
        nearest multiple of $2^{\Theta(K)}$, where $K=\Theta(k)$. Then we have for all $x \in [t]^n$:
$$\Pr\left[\left(1-2^{-K}\right)\cdot\normt{x}^2 \leq \normt{Rx} \text{ and } \forall_i (Rx)_i^2 < \frac{\normt{x}^2}{n}K\right] > 1 - neg(k).$$
\end{theorem}

We note that the only place where we use a PRG in our constructions,
is in order to be able to apply this theorem (in order to be able to
obtain shared randomness with low communication, our parties will
share a seed of a PRG that they will expand and use to sample).

\subsection{Defining Secure Computation for Distribution Testing}

Defining security for a distribution testing protocol requires some
care, due to two new features that do not come up in the
standard setting of secure computation of a function: first, this is
``testing'' and not ``computing'', and, second, the function of interest is
defined with respect to distributions, but the inputs that the parties
use in the computation are samples.
Before providing our formal definition, we %
discuss these issues and our choices.

A testing problem can be described as a partial boolean function $g$,
with the goal of computing $g(x)$ whenever $g$ is defined on the given
input $x$ (e.g, when the input consists of two distributions that are
either identical or $\epsilon$-far). 
If the input $x$ is such that $g(x)$ is not defined (we will refer to
this as an {\em input in the gray zone\/}), the property testing
definition (and literature) does not care about whether the output is
0 or 1. Indeed, this flexibility of having a gray zone with no
correctness requirement imposed on it is precisely what allows for more
efficient testing algorithms.

In contrast, for security purposes, we must care about what the
protocol outputs when the input is in the gray zone, as this may
reveal information about the inputs.  Specifically, standard secure
computation notions require that each party learns nothing beyond what
follows from their own input and output.  For a protocol testing some
property, when the output $g(x)$ is defined, we can (and will) follow
this paradigm.  When the input is in the gray zone and $g$ is not
defined, it is tempting to require that the protocol reveals this
fact and no other  information about the input.  However, it's easy
to see that such a protocol in fact is a secure computation of a
complete function (where for each input $x$ the protocol outputs
whether $g(x)=0$, $g(x)=1$, or $g(x)$ is undefined).
We could instead require that for inputs in the gray zone, the output
of the protocol is some distribution over $\{0,1\}$, independent of
the input.
But such a protocol again is in fact computing some
complete function of the input, which intuitively defeats the point of
using testing (versus computing) to gain efficiency. 
This intuition suggests that secure testing of any property (namely,
secure computation of a partial function) cannot be achieved with
better efficiency than secure computation of a (very related) complete
function.
We do not attempt to formalize (or refute) this intuition here. 

Instead, %
we show that, for our more special case of secure {\em
  distribution\/} testing, some information must be leaked in the gray
zone, {\em even if one disregards efficiency considerations}. Indeed, consider a
closeness testing setting where Alice and Bob inputs are $t$ samples
from $a$ and $b$ (respectively) over $[n]$ which are either
$(a,b)=(a_0,a_0)$, or $(a,b)=(a_0,b_\eps)$ for some distributions
$a_0,b_\eps$ with $\normo{a_0-b_\eps}=\epsilon$. Any correct protocols
must have different outputs on such instances with, say, probability
$0.99$. Now define $\hat{a} = \delta \cdot a'
+ (1-\delta)\cdot a$ for $\delta=0.0001/t$, where $a'$ is $a$ defined on
a new, unique set of letters;
Clearly, the distribution of $t$-sample inputs from $(a,b)$ and
$(\hat{a},\hat{b})$ are the same (for each instance), except with
$0.0001$ probability, and therefore the distribution over Alice and
Bob views (and in particular, the output) in the latter case must be
statistically close to the first one, and hence differ with at least
$0.98$ probability. However $(\hat{a},\hat{b})$ in both $(a_0,a_0)$
and $(a_0,b_\eps)$ instances are in the gray zone. It is therefore not
possible to simulate Alice's view of $\Pi$ in any statistically
significant manner without additional information (besides the mere
fact the distance is in the grey zone).

Following the above, we define a protocol to be a secure
computation of a testing task (or partial function) $g(x)$,  if it is  
a secure computation of some (complete, possibly randomized) function
$f(x)$, where 
$f(x)=g(x)$ whenever $g(x)$ is defined; when $g(x)$ is not defined,
$f(x)$ is some function of the input $x$.  
Thus, each party learns nothing beyond what follows from
their own input and the output $f(x)$ (which means there's some
leakage $f(x)$ in the gray zone, but no leakage beyond the testing
output when the testing promise holds). 
We will require $f$ to be a boolean function, so that the leakage in
the gray zone is at most one bit of information about the input (we
will also mention other possible generalizations of this
definition).
Finally, we note that for distribution testing, the function $g$ is
defined on the distributions, but the actual inputs of the parties are
samples. 
This is not a
major distinction for correctness, as with sufficiently many samples
we can typically obtain the correct result
with high probability.  However, for  security,
the leakage $f$ necessarily applies to the  samples (our only access
to the distributions), and not the distributions themselves. 

We are now ready to provide our definition of secure computation of
distribution testing.

\begin{definition} \label{def-secdisttest}
Let $D$ be a set of input distributions over $\bigtimes_{i=1}^d[n_i]$, and 
let $g: D \to \{0,1\}$ be a partial boolean function, defined on a
subset $P\subseteq D$ (with $g(p)=\bot$ whenever $p\in D\setminus
P$).

Let $\pi$ be a $d$-party protocol, and let $k$ be a security
parameter. We say that $\pi$ is a {\em $t$-sample secure distribution testing protocol} (for
the testing task defined by $g$), if
there exists a boolean function $f: \{\bigtimes_{i=1}^d[n_i]\}^t \to \{0,1\}$
such that the following holds: 

\paragraph{Correctness:}
for any $p\in P$,
$$\Pr_{\zeta_1\ldots \zeta_t \sim_\text{i.i.d } p}
[f(\zeta)=g(p)] = 1-neg(k) $$

\paragraph{Security:}
For any $\zeta \in \{\bigtimes_{i=1}^d[n_i]\}^t$,
if we give each player $i \in [d]$ the input $\left(1^k, \zeta_1(i),\ldots,\zeta_t(i)\right)$, then protocol $\pi$ is a
secure computation of the function $f(\zeta)$.

\end{definition}

We note that the security condition can be instantiated with any
standard 
secure computation notion.  In this paper we focus 
on the semi-honest model, although malicious security can be obtained
by a standard transformation with low communication overhead. 

\paragraph{More General Variants of Secure Distribution Testing.}
The definition can be naturally extended in various ways.
We may allow non-boolean $f$ (more bits of leakage), relax the
requirement that $f=g$ whenever $g$ is defined (allowing additional
leakage even if the promise holds), or impose restrictions on $f$
based on which leakage is deemed more or less
reasonable.\footnote{Such restrictions would likely be
  application-dependent and subjective, as a general theory comparing
  which leakage functions are qualitatively ``better'' is an open
  research area in cryptography.}   
We may also generalize the scope of distribution testing tasks modeled
by the definition (e.g., we may allow one party to get more samples
than the other, or give one party a complete description of a
distribution rather than samples).
For simplicity and ease of exposition, we do not develop
these extensions here, and stick with the simpler definition above.

\subsection{Secure Closeness Testing Protocol}

We now describe a secure distribution testing protocol that achieves 
the same communication complexity as the insecure protocol from
Section~\ref{sec_priv}, up to poly-logarithmic factors. Specifically, we show the following theorem:  

\begin{theorem}[Closeness, Secure]\label{priv_cl_tst}
Fix a security
parameter $k>1$. Fix $n>1, \eps\in (0,2)$, and let $t$ be such that $t\ge C \cdot k \cdot
\max\left(n^{2/3}\cdot \epsilon^{-4/3}, \sqrt{n}\cdot
\epsilon^{-2}\right)$ for some (universal) constant $C>0$. Then, assuming PRG and OT, there exists a {\em secure
  distribution testing protocol} for $\CT2p_{n,t,\epsilon}$ which uses
$\tilde{O}_k\left(\frac{n^2}{t^2\epsilon^4} + 1\right)$
communication.
\end{theorem}

The high level approach is similar to that of the previous protocol, 
in that we also estimate the squared $\ell_2$
distance between samples drawn from split distributions $a_S$ and
$b_S$, for some split-distribution set $S$, to distinguish between
$a=b$ and $\normo{a-b} \geq \epsilon$. One challenge in securing
the protocol is that we would like Alice and Bob to agree on an alphabet (and specifically, use the same split-distribution set $S$) without compromising their inputs. To address this,
Alice and Bob will run a secure computation over a circuit of size $\tilde{O}_{k,\epsilon}(n^2/t^2 + 1)$
which  simulates samples from the split distribution instead of communicating the set $S$. 
The main idea is to approximate the distance of the unsplitted
occurrences vectors $A$ and $B$ (representing samples drawn from $a$
and $b$), and \say{manually} add the difference between the unsplitted
and the splitted distance for each $i \in S$. That means, the
protocols will:
\begin{enumerate}
	\item Approximate $\normt{A - B}^2$;
	\item Add exactly $\normt{A_S - B_S}^2 - \normt{A - B}^2$;
\end{enumerate}
There is however an issue with this approach: $\normt{A - B}^2$ might
be much larger than $\normt{A_S - B_S}^2$, and therefore the summation
of (1) + (2) above might not be a good enough approximation of
$\normt{A_S - B_S}^2$. To overcome this issue, we first define the
notion of \say{capped} vectors as follows:
\begin{definition}[Capped Vectors]
	For $X\in \R^n$, we define $X' \in \R^n$ to be the 'capped vector of X with threshold $L$' iff $\forall i\in [n], X'_i=\min(L,X_i)$.
\end{definition}

Now, instead of approximating $\normt{A - B}^2$, we will approximate \say{capped} versions of $A, B$
for some carefully chosen threshold $L$ termed $A', B'$, and show the capped distance
is a good--enough approximation to the splitted distance $\normt{A_S - B_S}^2$. Our revised plan is therefore:

\begin{enumerate}
	\item Approximate $\normt{A' - B'}^2$;
	\item Add exactly $\normt{A_S - B_S}^2 - \normt{A' - B'}^2$;
\end{enumerate}

{\bf Securely Approximating $\normt{A' - B'}^2$.} For this task, we will use
the techniques of \cite{Indyk:2006:PPA:2180286.2180304}. Recall the
$\ell2-Approx$ protocol from \cite{Indyk:2006:PPA:2180286.2180304}
samples an orthonormal matrix with rounded entries $R$ as per Theorem~\ref{randmat}, and samples
sufficient coordinates from $Rx$. To ensure the protocol works
correctly, efficiently and securely for all distances, it scans all
possible distance magnitudes (termed $T$) starting with the largest
possible one and dividing by 2 each time, each time with a new
circuit. In our case, however, we need only to distinguish if an
approximation is more than some threshold, so it suffices to run one
ROM sub-circuit with $T$ preselected for such threshold.

{\bf Choosing a set $S$.} We also take a slightly different approach
for splitting the samples. While in
\cite{diakonikolas2016new} and Section~\ref{sec_priv},
we split the distributions using samples from $b$ only, and show it
suffices to upper bound the $\ell_2$ norm for only one of the
distributions we test by comparing $a$ and $b$ second moment
approximations, here we will upper bound both $a_S, b_S$ second
moments by constructing 2 multisets: $S_a$ from $a$ and $S_b$ from $b$
and setting $S \triangleq S_a \uplus S_b$ (where $\uplus$ denotes the sum of multiplicities of each element in $S_a$ and $S_b$). While not necessary for
correctness, this new technique turns out to both simplify our
protocol and improve its security guarantees.

We now show that $\normt{A'-B'}^2 = \tilde{O}(\normt{A_S - B_S}^2)$:

\begin{lemma}\label{lm_sampledist}
	Let $p$ be distribution on $[n]$, and let $X, Y$ be the occurrence vectors of $t_1, t_2$ samples drawn from $p$ where $t_2 \geq t_1$, then with high probability, we have that for all $i \in [n], X_i \leq 50\ln(n) \cdot \max(1, \tfrac{t_1}{t_2}Y_i)$.
\end{lemma}
\begin{proof}
	By the Chernoff bound, for each $i \in [n]$ we have either:
	\begin{enumerate}
		\item $p_i \leq 1/t_1$, and then $\Pr[X_i > 10\ln(n)] \leq 1/100n^2$; or
		\item $p_i > 1/t_1$, and then $\Pr[X_i > 10\ln(n)t_1 p_i] \leq 1/100n$ and $\Pr[Y_i < t_2 p_i/5] \leq 1/100n^2$.
	\end{enumerate}
	so the claim follows the union bound over all coordinates.
\end{proof}

\begin{corollary}\label{cor_splitcap}
	Let $A, B$ be the occurrence vectors of $t$ independent
        samples drawn from each $a, b$. Let $A', B'$ be capped vectors
        of $A, B$ with threshold $L \in [t]$. Let $S_a, S_b$ be
        multisets of independent samples drawn from $a, b$ of size
        $t/L$. Finally, for $S
        \triangleq S_a \uplus S_b$, we
        define $A_S, B_S$ to be the occurrence vectors of the $t$
        samples encoded into $A, B$, recasted as being drawn from $a_S,
        b_S$. Then, with high probability: $\normt{A'-B'}^2 \leq
        100\ln(n)\cdot\normt{A_S - B_S}^2$.
\end{corollary}
\begin{proof}
	For each letter $i \in [n]$, if both $A_i$ and $B_i$ are
        larger than $L$, then LHS is $0$. Hence, we will now assume
        that $B_i<L$ w.l.o.g.. 

	By Lemma \ref{lm_sampledist}, we have with high probability for all $i \in [n]$:
	\begin{enumerate}
		\item $\{i$ multiplicity in $S_b\} \leq 50\ln(n)$.
		\item $\{i$ multiplicity in $S_a\} \leq 50\ln(n)\max(1,A_i/L)$.
	\end{enumerate} 
	
	Note that on one hand, since $B'_i = B_i$ and $A'_i =
        \min(A_i,L)$ we have $\tfrac{(A'_i-B'_i)^2}{(A_i-B_i)^2} \leq
        \tfrac{(A'_i)^2}{(A_i)^2} \leq \tfrac{L}{A_i}$, and on the other hand, we have that $\normt{A_S(i,\cdot)-B_S(i,\cdot)}^2 \cdot \{i \text{ multiplicity in } S \} \geq (A_i-B_i)^2$. By combining both inequalities we obtain that with high probability, all letters satisfy the equation $(A'_i-B'_i)^2 \leq 100\ln(n)\cdot\normt{A_S(i,\cdot) - B_S(i,\cdot)}^2$ and the corollary follows the summing over all letters.
\end{proof}

{\bf Adding $\normt{A_S - B_S}^2 - \normt{A' - B'}^2$.}  We note that
simulating the splitted samples in $A_S$ and $B_S$ can be done
independently for each letter $i \in [n]$, as $A_S(i,\cdot)$ and
$B_S(i,\cdot)$ are (random) functions of the number of occurrences of
the letter $i$ in $S$, $A$ and $B$, independently of the other
letters. Furthermore, the vectors $A_S-B_S$ and $A'-B'$ differ in only
$O(|S|)$, plus $O(t/L)=O(|S|)$, coordinates. This fact allows Alice and Bob
to prepare simulated samples for all possible multisets in polynomial
offline time, and the secure circuit can look up the correct
simulation of $A_S(i,\cdot)$ and $B_S(i,\cdot)$ for each $i \in S$ and
calculate $\normt{A_S - B_S}^2 - \normt{A' - B'}^2$ using $O(|S|)$
such lookups.

We define a method for Alice and Bob to prepare a look-up table of
simulated samples from split distributions $a_S$ and $b_S$ for all
possible multisets $S$ of size at most $t$.

\begin{definition}
	Given some occurrences vector $X \in [t]^n$, we will define the 3D Split Occurrences Matrix $X^{SM} \in [t]^{n\times t \times t}$ as the following random process:
	for any $i \in [n], j \in [t]$, we will split $X_i$ into $j$
        buckets $X^{SM}(i,j,1),...,X^{SM}(i,j,j)$ by recasting each sample to a
        random bucket uniformly.
\end{definition}

We present our protocol next.

\fbox{%
\parbox{450pt}{
\textbf{Secure protocol $\Pi$ for the $\CT2p$ problem} \\
Let $\zeta$ be i.i.d. samples from the product distribution $p = a
\times b$.\\
\textit{Alice's Input:} $1^k$, first coordinates of
$\zeta_1,\ldots,\zeta_t$ \\
\textit{Bob's Input:} $1^k$, second coordinates of $\zeta_1,\ldots,\zeta_t$ \\
\\
\textit{Output:} 1 if $a=b$ and 0 if $\normo{a-b}\geq \epsilon$
\\
\begin{enumerate}
        \item Let $K=\Theta(k)$, and $c$ is some constant.
	\item Let $t'=t/K$; $L= \max(1,\frac{t'^3 \cdot
          \epsilon^4}{c \cdot n^2})$; $\alpha = \Omega(\sqrt{L/t'})$;
          $l=\Theta(k^2\ln^2(n)/\alpha^2)$.
	\item Alice and Bob exchange the seed of the PRG $G$, and
          generate matrix $R$ as in Theorem \ref{randmat}.
	\item Alice partitions the $t$ samples into $K$ Sample Sets of
          $t'$ samples each, and for each Sample Set prepares ROM
          entries $RA', A^{SM}, M_a, X_{S_a}$ as follows:
\label{prot:ssA}
	\begin{enumerate}
		\item Generate multi-set $S_a$ using $t'/2L$ samples from the Sample Set according to Lemma \ref{split_lemma}.
		\item Let $X_{S_a}$ be the occurrence vector of $S_a$ 
		\item Let $A$ be the occurrence vectors of another
                  $t'/2$ samples from the Sample Set.
		\item Let $A'$ be capped vector of $A$ with threshold $L$.
		\item Let $M_a = \{i : i \in S_a \vee A_i > L\}$
                \item Let $A^{SM}$ be 3D Split Occurrences Matrix of $A$.
	\end{enumerate}
	\item Bob similarly prepares ROM entries $RB', B^{SM}, M_b, X_{S_b}$.
\label{prot:ssB}
         \item Alice and Bob run a secure circuit with ROM, $C_\Pi$ to compute
           the following function.
           \begin{enumerate}
           \item For each of the sample sets $j\in [K]$:
	     \begin{enumerate}
		\item For $i \in M_a \cup M_b$
		\label{prot::comp_yi}
		\begin{enumerate}
			\item Let $m_i = 1 + X_{S_a}(i)  + X_{S_b}(i) $
			\item Let $y_i = \normt{A^{SM}(i,m_i,\cdot)-B^{SM}(i,m_i,\cdot)}^2 - (A'_i-B'_i)^2$.
		\end{enumerate}
		\item Compute $O_j = \sum_{i \in M_a \cup M_b} y_i$
		\item Let $T = 2(\tau - O_j)$, where $\tau$ is the threshold from Lemma \ref{app_dist_test}.
		\item Generate random $i_1,...,i_l \in [n]$ and compute $\left.d_1=R(A'-B')\right|_{i_1}^2,...,d_l=\left.R(A'-B')\right|_{i_l}^2$.
		\item Generate $z_1,...,z_l$ from independent
                  Bernoulli with biases $nd_1/TK,...,nd_l/TK$.
		\item Let $D_j = O_j + \tfrac{TK}{l}\sum_{i\in[l]} z_i$.
		\item If $D_j > \tau$ vote $0$; otherwise, vote $1$.
	\end{enumerate}
	   \item Output the majority of the votes from the $K$ tests. 
           \end{enumerate}
           \end{enumerate} 
}}\\

\subsection{Protocol analysis}
We now prove Theorem~\ref{priv_cl_tst} on the correctness and security
of the protocol $\Pi$.

\begin{proof}[Proof of Theorem \ref{priv_cl_tst}]
The proof proceeds in the following steps:
\begin{enumerate}
\item
First we define a boolean function $f(\zeta)$.%
\item
We show that for $p=a \times b$ such that $a=b$ or $\normo{a-b}\geq \epsilon$, the function $f(\zeta)=g(p)$, 
 whenever $\zeta \sim_{i.i.d} p$ except with negligible probability.
\item
Finally, we show that the protocol $\Pi$ is a secure computation of $f(\zeta)$.
\end{enumerate}

We define $f(\zeta)$ as follows: \\
\fbox{%
\parbox{450pt}{
\textbf{$f(\zeta)$} 
\begin{enumerate}
	\item Partition the samples into $K$ Sample Sets
          $\zeta^1,...,\zeta^K$ as per steps \ref{prot:ssA} and
          \ref{prot:ssB} of $\Pi$.
	\item Compute the majority of the $K$ test results, where a
          test does the following for each of the $K$ Sample Sets $\zeta^i$:
	\begin{enumerate}
		\item Generate multi-sets $S_a, S_b$ as in the protocol.
		\item Let $S = S_a \uplus S_b$.
		\item Generate $A, B$ and calculate $A', B', A_S, B_S$.
		\item Let $\Delta_1 = \normt{A_S-B_S}^2 - \normt{A'-B'}^2$.
		\item Let $T=2(\tau - \Delta_1)$.
		\item Generate $z_1,...,z_l$ from independent
                  Bernoulli with bias $\normt{A'-B'}^2/TK$.
		\item Let $\Delta_2 = \tfrac{TK}{l}\sum_{i=1}^l z_i$.
		\item If $\Delta_1 + \Delta_2 \geq \tau$ vote 0; otherwise vote 1.
	\end{enumerate}
\end{enumerate}
}}

We now show that, whenever $a=b$ or $\normo{a-b} \geq \epsilon$, we
obtain $\Pr_{\zeta_1\ldots \zeta_t \sim_\text{i.i.d. } p}
[f(\zeta)=g(p)] = 1-neg(k)$, for $p = a \times b$. We show
$f(\zeta^i)$ votes correctly for each Sample Set $\zeta^i$, with high
probability. First, one can see that if $\normt{A'-B'}^2 \geq T$ then
$E[z_i] \geq 1/K$ and, by the Chernoff bound for $l\ge \Omega(k^2)$, we
have that $\Delta_2 \geq \tau - \Delta_1$ with $1-neg(k)$
probability%
; the test vote is 0 for this Sample
Set. Similarly, if $\normt{A'-B'}^2 \leq T/4$, the test votes 1 with
$1-neg(k)$ %
 probability %

Otherwise, when $\normt{A'-B'}^2\in[T/4,T]$, we have that $\Delta_2$
is a $\left(1\pm\alpha/100\ln(n)\right)$-factor approximation of
$\normt{A'-B'}^2$ with $1-neg(k)$ probability, by the Chernoff
bound for $l=\Omega(k^2\cdot \ln^2n/\alpha^2)$.

We now invoke Corollary \ref{cor_splitcap}, and obtain with high probability that:
	\begin{align*}
		\left|\Delta_1+\Delta_2 - \normt{A_S - B_S}^2\right| &= \left|\Delta_2 - \normt{A'-B'}^2\right| \\
		&\leq \alpha/100\ln(n) \cdot \normt{A'-B'}^2 \\
		&\leq \alpha \cdot \normt{A_S - B_S}^2.
	\end{align*}
	
	It is left to show we have sufficient samples for the vote to
        be correct with high probability. For that we use Lemma
        \ref{split_lemma} to bound $\normt{a_S}$ and $\normt{b_S}$ as
        follows:

	\begin{enumerate}
		\item We note that $|S_a| = |S_b| = t'/2L$. Therefore,
                  $S_a$ and $S_b$ contain subsets $S'_a \subseteq S_a$
                  and $S'_b \subseteq S_b$ of size $\Poi(t'/4L)$ each,
                  with probability at least $1-2\tfrac{(t'/4L)^2}{t'/4L}=1-2t'/4L=1-o(1)$.
		\item Therefore,
                  $\E\left[\normt{a_{S'_a}}\right],\E\left[\normt{b_{S'_b}}\right]
                  \leq \sqrt{4L/t'}$ and we have both norms
                  $O(\sqrt{L/t'})$ with
                  probability 0.99, by the Markov bound.
		\item Since $S'_a, S'_b \subseteq S$, then
                  $\normt{a_S},\normt{b_S} = O(\sqrt{L/t'})$.  %
		\item Furthermore, as $t' = \Omega(\sqrt{n} /
                  \epsilon^2)$, we have $|S| = 2t'/L = O(n)$ by our choice of $L$. 
	\end{enumerate}

From now on we assume all the above.
According to Lemma \ref{app_dist_test} and Fact \ref{split_fact}, all we need to make the test successful with high probability is $O(\sqrt{L/t'}\cdot n / \epsilon^2) = O(t'/c + n/\sqrt{t'}\epsilon^2) = O(t'/c + n^{2/3}/\epsilon^{4/3})$ many samples (since $t'\geq O(n^{2/3}\epsilon^{-4/3}$). By choosing $c$ to be high enough constant, $t'/2$ samples suffice to make the vote correct with high probability. Therefore majority vote over $K$ sample sets amplify such probability to $1-neg(k)$. \\

We'll now show $\Pi$  is a secure computation for
$f(\zeta)$ for any $\zeta$. For correctness, we'll show that for all $\zeta$,
$\E[\Pi(\zeta)]-\E[f(\zeta)] = neg(k)$.
This is sufficient, since both $\Pi$ and $f$ provide an output in
$\{0,1\}$ (so the output expectation is the same as the probability
the output equals 1). 
We'll first show that for all sample
sets $\zeta^1,\ldots,\zeta^K : \Delta_1 \equiv O_j$.

\begin{lemma}\label{cl_lm1}
	Let $A_S, B_S$ denote the occurrences vector of simulated samples from $a_S, b_S$ respectively, then we have for all sample sets $\normt{A'-B'}^2 + O_j \equiv \normt{A_S-B_S}^2$
\end{lemma}
\begin{proof}

For any $i \in [n]$, let $m_i = 1+\{i$ multiplicity in $S\}$. We have either:
	\begin{enumerate}[(i)]
		\item $i \notin M_a \cup M_b$, and then:
		\begin{enumerate}
		 	\item $(A'_i-B'_i)^2 = (A_i-B_i)^2$ (since $A_i, B_i \leq L$).
		 	\item $\normt{A_S(i,\cdot) - B_S(i,\cdot)}^2 = (A_i-B_i)^2$ (since $i \notin S$).
		 	\item $i$ does not contribute to $O_j$.
		\end{enumerate}
		\item $i \in M_a \cup M_b$, and then:
		\begin{enumerate}
		 	\item $i$ contribution to $O_j$ is $y_i =
                          \normt{A^{SM}(i,m_i,\cdot)-B^{SM}(i,m_i,\cdot)}^2 -
                          (A'_i - B'_i)^2$, and 
                          $$\normt{A^{SM}(i,m_i,\cdot)-B^{SM}(i,m_i,\cdot)}^2
                          \equiv \normt{A_S(i,\cdot) - B_S(i,\cdot)}^2.$$
		\end{enumerate}
	\end{enumerate}
	Therefore, each $i$ contribution to LHS and RHS are equivalent. This concludes the proof.
\end{proof}

So we are left to show that $\Pr[D_j \geq \tau] - \Pr[\Delta_1 + \Delta_2 \geq \tau] = neg(k)$. If we have $\normt{A'-B'} > T$, then both events occurring with prob. $1-neg(k)$. Otherwise, we use Theorem \ref{randmat} and obtain that:
\begin{enumerate}
	\item With prob. $1-neg(k)$, $\forall i (RA'_i-RB'_i)^2 \leq TK/n$.
	\item $(1-2^{-\Theta(K)})\normt{A'-B'}^2 \leq E_i[n(RA'_i-RB'_i)^2] \leq \normt{A'-B'}^2$.
\end{enumerate}
Therefore, by linearity of expectations and the union bound we obtain that the distribution over the $z_i$ differ by at most $neg(k)$ between $f(\zeta)$ and $\Pi(\zeta)$, and therefore the probability to output 0 differ by at most $neg(k)$.\\

Next, we show  that the protocol $\Pi$ is a secure computation of
$f$. We note that
the only communication in the protocol is sending is the random seed,
followed by the secure circuit computation of the final output,
which is indistinguishable from $f(\zeta)$. Moreover, as described above, either $\Pr[f(\zeta) = \Pi(\zeta)] \geq 1-neg(k)$ or the distribution over the $z_i$ (for a given seed) differ by at most $neg(k)$. Thus, replacing the 
secure circuit $C_\Pi$ with an oracle computing the output, the
simulator (for either Alice or for Bob) outputs its input, a random seed for $G$, and the final output, to
generate a distribution that is statistically close to that in the
real-world.  \\

Finally, we analyze the communication complexity of the protocol $\Pi$
for $\CT2p$.  
To analyze the circuit size of $C_\Pi$, one can observe that:
 \begin{itemize}
 	\item Computing $O_j$ can be done with a circuit of size
          $\tilde{O}(|M_a| + |M_b|) = \tilde{O}(t' / L) =
          \tilde{O}(\tfrac{n^2}{\epsilon^4\cdot t^2} + 1)$ (by computing each $y_i$ in step \ref{prot::comp_yi} using $\tilde{O}(1)$ computations).
 	\item Next, computing $D_j$ (with $O_j$ as input) can be done with a 
          circuit of size $\tilde{O}(l) = \tilde{O}_k(\log^2n
          \cdot 1/\alpha^2) = \tilde{O}_k(t' / L) = \tilde{O}_k(\tfrac{n^2}{t^2\epsilon^4} + 1)$.
    \item Computing majority over $K$ (sub-)circuits adds multiplicative factor of $\Theta(k)$ to the circuit size.
 \end{itemize}
  Thus $C_\Pi$ is of size $\tilde{O}_k\left(\frac{n^2}{t^2\epsilon^4}
  + 1\right)$,
  and the ROM consists of $s \triangleq \text{poly}(n,t,1/\epsilon,k)$ words of size $r \triangleq \tilde{O}(1)$ each.
  Therefore, the communication complexity of the secure computation of
  $C_\Pi$ is $|C_\Pi| \cdot T(r,s)$, where $T(r,s)$ is is the
  communication complexity of 1-out-of-$s$ OT on words of size $r$.
  Finally, the total communication of the protocol $\Pi$ additionally
  also includes the   length of the seed for $G$.

  Plugging in our OT and PRG assumptions, we obtain the required
  communication bound  (using Theorem 
  \ref{thm_nn01}).
\end{proof}

\section{Independence Testing: Communication and Security}
\label{sec:independence}

In this section, we present our protocol for independence testing. To
streamline the presentation, we give directly a secure protocol,
noting that obtaining communication-efficiency is the main challenge
for this problem. Overall, we prove the following theorem:
\begin{theorem}[Independence, Secure]
\label{thm_2pit}
Fix a security parameter $k>1$. Fix $\eps\in (0,2)$, $1\le m\le n$, and let $t$ be such that $t
\geq C \cdot k \cdot \left(n^{2/3}m^{1/3}\epsilon^{-4/3} +
\sqrt{nm}/\epsilon^2\right)$, for some (universal) constant $C>0$, and assuming
OT cryptographic assumption, there is a secure 
distribution testing protocol for $\IT2p_{n,m,t,\epsilon}$ using
$\tilde{O}_{k}\left(\frac{n^2\cdot m}{t^2\epsilon^4} + \frac{n\cdot
  m}{t \epsilon^4} + \frac{\sqrt{m}}{\epsilon^3}\right)$ bits of communication.
\end{theorem}

Note that similarly to $\CT2p$, our protocol for $\IT2p$ obtains
better communication than via the straight-forward approach of Alice
sending its input to Bob, and Bob invoking standard independence
testing algorithm. We show a matching lower bound on communication
for $t\to \infty$ in Section~\ref{sec:indepLB}.

We start by providing some intuition behind our protocol, setting up
some definitions along the way.  As before, our protocol for $\IT2p$
uses the framework of splitting the distributions to reduce the
problem to testing an $\ell_2$ distance. While it may seem tempting to
reuse the sketching technique to estimate the $\ell_2$ distance
between the joint distribution and the product of its marginals, this
technique does not help here as this distance cannot be approximated
without first combining the samples from Alice and Bob. To overcome
this inherent obstacle, we use an {\em alphabet reduction}
technique. At a high level, our reduction samples a random rectangle
$R$ of the underlying split distribution alphabet, and tests its
distance to the marginal product distribution. This concept helps us
to improve communication in two ways. First, we obtain better
communication as we need to deal with a smaller alphabet problem, and,
second, we can use the additional samples for tighter $\ell_2$ bound on
the smaller alphabet product distribution. Each of these two
improvements achieve linear improvement with more samples and
therefore we hope for a quadratic improvement overall. The latter
improvement is only effective while $t \ll n$ samples, and hence the
quadratic improvement kicks in only in that regime.

To formalize the alphabet-reduction idea, we start by making the
following definition:

\begin{definition}
Fix some distribution $p$ over alphabet $[n]$. For $\UC \subseteq [n]$
let $p_{\rvert \UC}$ be the distribution $p$ conditioned on the set
$\UC$: i.e., $p_{\rvert \UC}$ is a distribution over alphabet $\UC$
such that for any $i \in \UC$, $p_{\rvert \UC}(j) = p(j)/\sum_{i \in
  \UC}p(i)$.
\end{definition}

Given joint distribution $p=(a,b)$, let $S_a, S_b$ be the multisets in $[n],
[m]$ of samples from $a, b$ respectively. Define the multiset
$S\subset [n] \times [m]$ to be such that $1 +
\{$multiplicity of $(i,j)$ in $S\} = (1 + \{$multiplicity of $i$ in
$S_a\})\cdot(1 + \{$multiplicity of $j$ in $S_b\})$.

Our idea is to sample a rectangle $R = R_a \times R_b \subseteq [n + |S_a|]
\times [m + |S_b|]$, and test the closeness of the sub-distributions
$\hat{p} \overset{\Delta}{=} (p_S)_{\rvert R}$ and $\hat{q} \overset{\Delta}{=} (a_{S_a})_{\rvert R_a} \times
(b_{S_b})_{\rvert R_b}$. We would like $R$ to satisfy the following
properties:
\begin{enumerate}
\item If $a$ is independent of $b$, denoted $a \independent b$, then $\hat{p}=\hat{q}$, and if $a$ and $b$ are $\eps$-far from being independent, then $\hat{p}$
  is at distance $\Omega(\epsilon)$ from $\hat{q}$. %
\item $\Pr_{i \sim p}[i \in R]$ is high enough, so that we have sufficient samples to test $\hat{p}$ and $\hat{q}$ for closeness.
\item $R$ is sufficiently small to allow small communication.
\item The secure circuit can simulate samples from $\hat{p}, \hat{q}$ using $\tilde{O}(1)$ ROM look-ups per sample.
\end{enumerate}

We obtain conditions (1) and (2) above by bounding the $\ell_2$ norm
of the underlying product distribution through splitting, and use such
bound also as a bound over the variance of the $\ell_1$ mass of $\hat{p}$
and its distance from $\hat{q}$. We analyze such
approach in Section \ref{subsec:ind_reduc}. For condition (4), we only
sample $R_a \subseteq [n + |S_a|]$ (assuming w.l.o.g. $n \geq m$), and
have $R_b$ be the complete alphabet of Bob $[m + |S_b|]$. The final
communication/sample trade-offs are obtained from the best trade-off
between conditions (2) and (3).

In order to be able to look-up and \say{pair} Alice and Bob samples
from some rectangle $R$, we need a method for storing a table of all
indices for each letter. For this, we use the following definition:
\begin{definition}
Given $t$ indexed samples $X \in [n]^t$, we define {\em the Indices
  Set Vector} $\mathcal{I} \in (2^{[t]})^n$ such that $\mathcal{I}(j)
= \{i \in [t]: X_i=j\}$.
\end{definition}

\subsection{Full protocol description}

In contrast to closeness testing, the main challenge for the $\IT2p$
problem is designing a communication-efficient protocol, while adding
security introduces relatively minor nuances. Hence, for ease of
exposition, we present the secure version of our protocol directly. 

Our overall protocol proceeds as follows at a high level. Alice and
Bob first generate multi-sets $S_a, S_b$. Alice then simulates one set
of indexed samples from $a_{S_a}$ (termed $A$) while Bob simulates
{\em two} sets of samples from $b_{S_b}$. First set (termed $B_p$) is
obtained using the sample set corresponding to the set $A$, and the
second one using a fresh independent sample set (termed $B_q$). Alice
then prepares Indices Set Vector $\mathcal{I}$ of $A$. Next, a secure
circuit with ROM samples a uniform subset $\UC$ from $[n + |S_a|]$,
looks up all of Alice's samples coming from $\UC$, and generates 2
joint sample sets: the first one by pairing each of Alice's samples
$A$ with the corresponding sample in $B_p$ (thereby simulating samples
from $\hat{p}$ as defined above), and then by pairing each of Alice's
samples $A$ with an independent sample from $B_q$ (thereby simulating
samples from $\hat{q}$). Finally, those two sets are then being tested
(directly) for closeness using Lemma \ref{app_dist_test}.

We note that the set $\UC$, together with the protocol's output can
possibly compromise Bob's data, and therefore we need the secure
circuit to compute $\UC$ (rather than Alice providing it to the
circuit). One challenge with such approach though, is the set $\UC$
may be larger than the communication bound. Luckily, we only care
about non-empty letters in $\UC$ which are (in expectation) of small
size, and since we sample $\UC$ uniformly from $[n + |S_a|]$, we can
overcome the issue by sampling the non-empty letters of $\UC$ from the
non-empty letters of $[n + |S_a|]$ (e.g. appears at least once in
$A$). For this we need to define a distribution over the size of
intersection of subsets chosen uniformly:

\begin{definition}
	For $n \geq 1$, and $\alpha, \beta \leq n$, we define a
        discrete distribution called Uniform Subset Intersection
        distribution $\mu(n,\alpha,\beta)$ as follows: Fix $S$ to be a
        set of size $n$, and $S_1 \subseteq S$ of size $\alpha$. Then
        $\mu(n,\alpha,\beta) \overset{\Delta}{\equiv} |S_1 \cap S_2|$,
        where $S_2 \subseteq S$ of size $\beta$ is chosen uniformly at
        random. 
\end{definition}

\begin{claim}\label{cl_usi}
	Let $S$ be a finite set of size $n \geq 1$, and let $S_1
        \subseteq S$ of size $\alpha$. Let $S_2$ be a uniform subset
        of $S$ of size $\beta$. Finally, let $S_3$ be uniform subset
        of $S_1$ of size $\mu(n,\alpha,\beta)$. Then $S_3 \equiv S_1
        \cap S_2$.
\end{claim}
\begin{proof}
One observe that $|S_3| = \mu(n,\alpha,\beta) \equiv |S_1 \cap S_2|$ by our definition of $\mu$. Furthermore, since $S_2$ was chosen uniformly at random, then each element of $S_1$ has equal and independent probability to be part of $S_1 \cap S_2$, and therefore $S_1 \cap S_2$ is also a uniform subset of $S_1$.
\end{proof}

We now present our protocol for $\IT2p$.  As before, our protocol uses
a secure circuit with ROM which will sample from Alice and Bob's input
strings. 

\fbox{%
\parbox{450pt}{
\textbf{Secure protocol $\Pi$ for the $\IT2p$ problem} \\
Let $\zeta_1,\ldots \zeta_t$ be i.i.d. samples from the joint distribution $p = (a,b)$.\\
\textit{Alice's Input:} $1^k$, first coordinates of
$\zeta_1,\ldots,\zeta_t$ \\
\textit{Bob's Input:} $1^k$, second coordinates of $\zeta_1,\ldots,\zeta_t$ \\
\\
\textit{Output:} 1 if $a \independent b$ and 0 if $p$ is $\epsilon$-far from any product distribution.

\begin{enumerate}
	\item Let $K=\Theta(k);\epsilon' = \Omega(\epsilon); t'=\min\{t/3K,O(n \cdot \sqrt{m}/\epsilon)\}; l = O(\max\{\tfrac{n^3 \cdot m}{t'^3\epsilon^4},\tfrac{n^2 \cdot m}{t'^2\epsilon^4},\tfrac{1}{\epsilon^2}\})$.
	\item Alice and Bob partition their samples into $3K$ Sample Sets of $t'$ samples each with corresponding indices termed $A_1,\ldots,A_{3K}, B_1,\ldots,B_{3K}$.
	\label{prot:ind_part}
	\item For $i \in \{0,\ldots,K-1\}$ Alice prepares ROM entries
          termed $A^i, \mathcal{I}^i, \Gamma^i$ as follows:
	\begin{enumerate}
		\item Generate multi-set $S_a^i$ using $\min\{t',n\}$ samples from $A_{3i+1}$ according to Lemma \ref{split_lemma}.
		\item Let $A^i$ be $A_{3i+2}$ recasted as being drawn from $a_{S_a^i}$.
		\item Let $\mathcal{I}^i$ be the Indices Set Vector of $A^i$.
		\item Let $\Gamma^i = \{$non-empty letters of $\mathcal{I}^i\}$. 
	\end{enumerate}
		\item For $i \in \{0,\ldots,K-1\}$ Bob prepares ROM
                  entries termed $B^i_p, B^i_q$  as follows:
	\begin{enumerate}
		\item Generate multi-set $S_b^i$ using $m$ samples from $B_{3i+1}$ according to Lemma \ref{split_lemma}.
		\item Let $B^i_p$ be $B_{3i+2}$ recasted as being drawn from $b_{S_b^i}$.
		\item Let $B^i_q$ be $B_{3i+3}$ recasted as being drawn from $b_{S_b^i}$. 
	\end{enumerate}
	\item A secure circuit with ROM($A^i,B^i_p, B^i_q,
          \mathcal{I}^i, \Gamma^i$) computes for each $i \in
          \{0,\ldots,K-1\}$ the following vote, and then outputs the
          majority vote:
	\begin{enumerate}
		\item Compute $\lambda = \min\{\mu(n+|S_a^i|,|\Gamma^i|,l),100t'l/n\}$
		\item Sample $\UC'$ uniformly from $\Gamma^i$ of size $\lambda$.
		\item Let $I^i = \bigcup_{j \in \UC'} \mathcal{I}^i(j)$.
		\item Compute $I^i_p, I^i_q, J^i_p, J^i_q$ as uniform disjoint subsets of $I^i$ of size $\min\{O(t'l/n),\lfloor |I^i|/4 \rfloor\}$ each.
		\item Compute $X_1^i = \{(A^i(j),B^i_p(j)) : j \in I^i_p \}$.
		\item Compute $X_2^i = \{(A^i(j),B^i_p(j)) : j \in J^i_p \}$.
		\item Compute $Y_1^i = \{(A^i(j),B^i_q(j)) : j \in I^i_q \}$.
		\item Compute $Y_2^i = \{(A^i(j),B^i_q(j)) : j \in J^i_q \}$.
		\item Count collisions in $X_1^i, Y_1^i$ and produce estimations of $\normt{\hat{p}}, \normt{\hat{q}}$ up to factor 2. Compute $\chi^i = 1$ if they these estimations agree up to factor 4 and $\chi^i = 0$ otherwise.
		\item Compute $\Delta^i = \normt{X-Y}^2$ where $X$ and $Y$ are the occurrence vectors of $X_2^i$ and $Y_2^i$.
		\item Vote $\chi^i \wedge (\Delta^i \leq \tau)$, where $\tau$ is threshold from Lemma \ref{app_dist_test}
	\end{enumerate}
\end{enumerate}
}}\\

\begin{remark}
\label{rem:it1way}
We
note that one can obtain a non-secure, 1-round 1-way communication
protocol by having Alice computing the set $\UC$ of size $l$(sampled from $[n + |S^i_a|]$), and sending to Bob
sufficiently many samples coming from $\UC$, along with corresponding indices
of such samples, while Bob performing the pairing and testing for
closeness as per steps $5.(e),\ldots,5.(k)$ above.
\end{remark}

\subsection{Analysis setup: alphabet reduction}\label{subsec:ind_reduc}
Before proceeding to the full protocol analysis, we first argue some
auxiliary lemmas, regarding reducing the larger independence testing
problem to a smaller closeness problem which requires less samples and
therefore smaller communication than the original problem.  The lemmas
state show that sampling sub-distributions from distributions with
bounded $\ell_2$ norms preserve some important properties.

We first show that if we sample sufficiently large sub-distribution
$p_{\rvert \UC}$ uniformly from some distribution $p$ with bounded
$\normt{p}^2$, then with high probability, $p_{\rvert \UC}$ has
density of the same magnitude as the fraction of the alphabet we are
sampling from, and $\normt{p_{\rvert \UC}}^2$ is similarly bounded.

\begin{lemma}\label{lem_sumpi}
	Let $p$ be some distribution over $[n]$ such that $\normt{p}^2 \leq U$, and let $\UC$ of size $l$ be a uniformly random subset of $[n]$. For $l\geq 100\cdot U\cdot n$ we have with 0.95 probability:
	\begin{enumerate}
		\item $\sum_{i \in \UC} p_i = \Theta(l/n)$;
		\item $\normt{p_{\rvert \UC}}^2 = O(U \cdot n/l)$.
	\end{enumerate}	
\end{lemma}
\begin{proof}
Note that the expectation of (1) is $l/n$ and its variance is at most $U\cdot l/n$ (roughly). Therefore, by the Chebyshev inequality and the lower bound on $l$, we have (1) with probability $\ge 0.99$. \\
For (2), we note that $\normt{p_{\rvert U}}^2 =\sum_{i \in \UC} p_i^2/ (\sum_{i \in \UC} p_i)^2$. The numerator is $O(U \cdot l/n)$ in
expectation, and therefore for high enough constant $O(U \cdot l/n)$ with
0.99 probability. The denominator is $\Omega(l^2/n^2)$ by the
above argument.
\end{proof}

We now derive by a similar argument that such sampling process also preserves distances.
\begin{lemma}\label{lem_od}
Let $p$ be some distribution over $[n]$ such that $\normt{p}^2 \leq U$, and let $\Delta \in [0,2]^n$. Let $\UC$ of size $l$ be a
uniformly random subset of $[n]$. If $\langle p,\Delta\rangle \geq
\epsilon$ and $l \geq O(U\cdot n / \epsilon^2)$, then we have with 0.9
probability that:
$$\frac{\sum_{i \in \UC} p_i \cdot \Delta_i}{\sum_{i \in \UC} p_i} = \Omega(\epsilon).$$
\end{lemma}

\begin{proof}
By Lemma \ref{lem_sumpi}, the denominator is $\Theta(l/n)$ with 0.95 probability. %
The expectation of the numerator is $l\epsilon/n$ and its variance is at most $U\cdot l/n$. Therefore, for $l \geq O(U\cdot n / \epsilon^2)$ (using high enough constant), the numerator is $\Theta(l\epsilon/n)$ with probability 0.95 by the Chebyshev inequality.
\end{proof}

Finally, we show that we can apply the above lemmas to obtain a reduction from testing a distribution $p = [n] \times [m]$, to testing closeness of a smaller sub-distribution (where we down-sample the letters from $[n]$ and
condition on those) to the product distribution. %
\begin{lemma}\label{lm_epsfar}
	Let $p$ be a distribution on $[n] \times [m]$ and let $p_1,
        p_2$ be its marginals. Fix $\epsilon < 2$. Let $\UC\subset [n]$ be a random subset of size $l$, such that $l \geq O(U\cdot n / \epsilon^2)$ where $U=\normt{p_1}^2$. Define $\hat{p} \triangleq p_{\rvert \UC \times [m]}$ and $\hat{q} \triangleq p_{1\rvert \UC} \times p_2$. There exists $\epsilon' = \Omega(\epsilon)$ such that with 0.9 probability,
        \begin{enumerate}
        	\item If $p$ is a product distribution, then $\hat{p} = \hat{q}$.
        	\item If $p$ is
        $\epsilon$-far from any product distribution, then $\normo{\hat{p} - \hat{q}} \geq \epsilon'$.
        \end{enumerate}
\end{lemma}

\begin{proof}%
	We obtain (1) immediately from the fact $p_{1\rvert \UC}$ and
        $p_2$ are the marginals of $\hat{p}$.

	For $i \in [n]$, let $\Delta_i = \normo{p(i,*) - p_2}$, where
        $p(i,*)$ is the right-marginal of $p$ conditioned on
        $p_1=i$. Since $p(i,j) = p(i,*)(j)\cdot p_1(i)$, we have that
        $\langle p_1,\Delta \rangle = \normo{p - p_1 \times p_2}$.  If
        $p$ is $\epsilon$-far from the product distribution, then
        $\langle p_1,\Delta \rangle \geq \epsilon$. We now invoke
        Lemma \ref{lem_od} using $p=p_1$, and get that with 0.9
        probability:
	$$\normo{p_{\rvert \UC \times [m]} - p_{1\rvert \UC} \times p_2} = \frac{\sum_{i \in \UC} p_1(i) \cdot \Delta_i}{\sum_{i \in \UC} p_1(i)} = \Omega(\epsilon).$$
	\end{proof}

\subsection{Analysis of the protocol}

\begin{proof}[Proof of Theorem \ref{thm_2pit}]
The proof proceeds in the following steps:
\begin{enumerate}
\item
First we define a boolean function $f(\zeta)$.%
\item
We show that for $p=(a,b)$ such that $a \independent b$ or $p$ is $\epsilon$-far from any product distribution, the function $f(\zeta)=g(p)$, 
 whenever $\zeta \sim_{i.i.d} p$ except with negligible probability.
\item
Finally, we show that the protocol $\Pi$ is a secure computation of $f(\zeta)$.
\end{enumerate}

For $f(\zeta)$, we define $f$ to be the boolean function producing a
bit by simulating steps $(1),\ldots,(5)$ of $\Pi(\zeta)$ by Alice and
Bob.  Indeed, such simplification is possible for $\IT2p$ since the
protocol involves no communication outside of the secure circuit
computation.

We now show that, whenever $a \independent b$ or $p$ is $\epsilon$-far from any product distribution, we
obtain $\Pr_{\zeta_1\ldots \zeta_t \sim_\text{i.i.d. } p}
[f(\zeta)=g(p)] = 1-neg(k)$. We prove for each $i \in \{0,\ldots,K-1\}$,  all the below happen with some high constant probability:
\begin{enumerate}
	\item $X_1^i, X_2^i$ are distributed as being drawn from $\hat{p}$ and $Y_1^i, Y_2^i$ are distributed as being drawn from $\hat{q}$.
	\item $\hat{p} = \hat{q}$ if $a \independent b$ and $\normo{\hat{p} - \hat{q}} \geq \epsilon'$ if $(a,b)$ are $\epsilon$-far from the product distribution.
	\item $X_1^i, Y_1^i$ contain sufficiently many samples to approximate $\normt{\hat{p}}, \normt{\hat{q}}$, and $X_2^i$ and $Y_2^i$ contain sufficiently many samples to test the closeness of $\hat{p}$ and $\hat{q}$ according to Lemma \ref{app_dist_test}.
\end{enumerate}
For (1), we note that $|\Gamma^i| \leq t'$ and therefore $\E[\mu(n+|S_a|,|\Gamma^i|,l)] \leq t'\cdot l/n$ and we have $\lambda = \mu(n+|S_a|,|\Gamma^i|,l)$ with 0.99 probability (by Markov). According to Claim \ref{cl_usi}, the process in which the circuit samples $\UC'$ from $\Gamma^i$ simulates the process of sampling $\UC$ of size $l$ uniformly from $[n + |S_a|]$ with $\UC' \equiv \Gamma^i \cap \UC$, and therefore $X^i$ and $Y^i$ are distributed as being drawn from $\hat{p}$ and $\hat{q}$. \\

We obtain (2) with 0.9 probability immediately from Lemma \ref{lm_epsfar}.\\

For (3), let us denote $t_{\textbf{MIN}}$ to be the samples required both to meet the conditions of Lemma \ref{app_dist_test}, and to produce a constant factor approximation for $\normt{\hat{p}}$ and $\normt{\hat{q}}$. Let us also denote $t_{\textbf{REAL}}$ to be the actual samples our protocol produces from $\hat{p}, \hat{q}$ in steps $5.(d),\ldots,5.(h)$ above.
Let $\gamma = \min\{t',n\}$. With $1-o(1)$ probability, $S_a, S_b$
contain subsets $S'_a \subseteq S_a$ and $S'_b \subseteq S_b$ of sizes
$\Poi(\gamma/2)$ and $\Poi(m/2)$ respectively. Therefore, according to Lemma \ref{split_lemma} and Lemma \ref{lem_sumpi}:
$$\normt{\hat{q}}^2 = \normt{a_{S_a \rvert \UC}}^2 \cdot \normt{b_{S_b}}^2 = O(\frac{n}{\gamma l}) \cdot O(\frac{1}{m}) = O(\frac{n}{\gamma l m})$$
And therefore:
$$t_{\textbf{MIN}} = O\left(lm\sqrt{\frac{n}{\gamma l m}}\cdot \epsilon^{-2} + \sqrt{lm}\right) = O\left(\sqrt{\frac{n l m}{\gamma \epsilon^4}}\right)$$
On the other hand we have by Lemma \ref{lem_sumpi} with probability 0.95:
$$t_{\textbf{REAL}} = \Omega\left(\frac{t' l}{n}\right)$$
As a result, there exists some constant $c$ for which we obtain $t_{\textbf{REAL}} \geq t_{\textbf{MIN}}$ with probability $0.95-o(1)$ whenever $l \geq \tfrac{c n^3 m}{t'^2 \gamma \epsilon^4} = c \cdot \max\{\tfrac{n^3 \cdot m}{t'^3\epsilon^4},\tfrac{n^2 \cdot m}{t'^2\epsilon^4}\}$.

We note as well that $\Delta^i$ is an exact computation of $\normt{X-Y}^2$ and therefore can also be considered $(1+\alpha)$ approximation of such quantity (for any $\alpha$).

Last, we need to show $l$ is properly defined, and in particular, $l \leq n + |S_a|$. Note that from the conditions of Theorem \ref{thm_2pit}, $t' \geq C \cdot O(\sqrt{nm}/\epsilon^2)$. We thus set $C$ to be a large enough constant such that $l \leq n$ holds.

To sum up, we have shown (by the union bound) that for each $i \in \{0,\ldots,K-1\}$, the conditions of Lemma \ref{app_dist_test} are met with probability $0.85 - o(1)$, and therefore $f(\zeta)$ votes correctly with probability $1/2 + \Omega(1)$. The final majority output over $K$ votes is therefore correct with $1-neg(k)$ probability as needed.

We'll now show $\Pi$ is a secure computation for
$f(\zeta)$ for any $\zeta$. For correctness, we need show that for all $\zeta$,
$\E[\Pi(\zeta)] - \E[f(\zeta)] = neg(k)$. In fact, those are equal by our definition of $f$ above.

Security follows immediately from the security of the secure circuit
with ROM, as there is no additional communication or randomness in the
protocol.

We now analyze the communication complexity of $\Pi$. We note all communication is invoked from the secure circuit in step (5) above. For each of the $K$ (sub-)circuits producing a vote, the circuit first makes $O(1)$ look-ups for each letter in $\UC'$. The circuit then calculates the number of corresponding sample indices in $\mathcal{I}^i$ and samples $O(t'l/n)$ such indices. One can observe that such sampling process is possible using a circuit of size $O(t'l/n)$ even with the (low-probability) scenario where there are $\omega(t'l/n)$ such indices since the index sets are disjoint. For each sampled index (in $I^i_p,I^i_q,J^i_p,J^i_q$), the circuit invokes $\tilde{O}(1)$ calculations to produce collision count and occurrence vector squared distance. Therefore, the circuit size is bounded by $\tilde{O}(\lambda + t'l/n) = \tilde{O}(t'l/n) = \tilde{O}\left(\frac{n^2\cdot m}{t'^2\epsilon^4} + \frac{n\cdot
  m}{t'\epsilon^4} + \frac{\sqrt{m}}{\epsilon^3}\right)$, and the overall circuit producing majority vote is just $K$ times that: 
  $$\tilde{O}_k\left(\frac{n^2\cdot m}{t^2\epsilon^4} + \frac{n\cdot
  m}{t\epsilon^4} + \frac{\sqrt{m}}{\epsilon^3}\right)$$
z

By assuming $OT$, and invoking Theorem \ref{thm_nn01}, we obtain the claimed complexity.
\end{proof}

\section{Independence Testing: Communication Lower Bound}
\label{sec:indepLB}

We now show our lower bounds for $\IT2p$. We show that any one-way
protocol requires $\Omega(\sqrt{m})$ bits of communication (regardless
of the security properties).

The lower bound result is obtained by reduction from the Boolean
Hidden Hypermatching (BHH) problem from
\cite{doi:10.1137/1.9781611973082.2}. In the BHH problem, Alice is
given $x \in \zo^n$, and Bob is given some complete matching $M$ of $[n]$
such that for all pairs $(i,j)$ in the matching: $i \xor j = b \in
\zo$. Bob's goal to output $b$ correctly with $2/3$ probability. For this
problem, \cite{doi:10.1137/1.9781611973082.2} show that
one-way communication complexity is $\Omega(\sqrt{n})$.\\

Using the above result, we now show the following:
\begin{theorem}\label{thm_lb_it}
For $n, t \in \N$, any one-way protocol for $\IT2p_{n,n,t,1,1/3}$
requires $\Omega(\sqrt{n})$ bits of communication.
\end{theorem}

\begin{proof}[Proof of Theorem \ref{thm_lb_it}]
	We show how, for given a BHH instance $p^b=(x,M)$, we can reduce it
        to the $\IT2p_{n,n,t,1}$ problem. Let $r \in [n]^t$ be some common uniformly random string. In addition, for $b \in \zo$, let $X_b = \{i \in [n] : x_i=b\}$. For each $i \in [t]$, Alice, Bob each will generate an indexed sample $(A_i, B_i)$ respectively as follows:
	\begin{enumerate}
		\item Let $b' = x_{r(i)}$; Alice samples $A_i$ uniformly from $X_{b'}$.
		\item Bob samples $B_i$ uniformly between $r(i)$ and
                  its pair in the matching $M$.
	\end{enumerate}
	One can observe that for any $i \in [t], j \in [n]$,
        $\Pr[A_i=j] = \Pr[B_i=j] = 1/n$, i.e., each of $A, B$ is
        distributed uniformly at random from $[n]$. Thus, if the
        distributions of $A$ and $B$ are to be independent, it must be
        a distribution uniform over $[n] \times [n]$.

Now let us analyze the probability distribution of each pair
$(A_i,B_i)$: for any $i \in [t], j \in [n], k \in [n]$, we have $(A_i,
B_i) = (j,k)$ iff the following precise conditions hold:
	\begin{enumerate}
		\item $x_{r(i)} = x_j$.
		\item $r(i)$ is either $k$ or its pair in the matching
                  $M$.
		\item $j$ was sampled by Alice.
		\item $k$ was sampled by Bob.
	\end{enumerate} 
	Events (3) and (4) are independent events which happen with
        probability $2/n$ and $1/2$ respectively conditioned on the
        first 2 events occurring.

For a $p^1$ instance (ie, when the output $b=1$), the probability of
the first 2 events occurring is $1/n$ for any $(j,k)$ (as there is a
single letter in $[n]$ that supports both conditions), and therefore
for any $i \in [t], j \in [n], k \in [n]$, we have $\Pr[(A_i, B_i) =
  (j,k)] = 1/n^2$. Hence $(A,B)$ are distributed as a product
distribution.

On the other hand, for a $p^0$ instance, the probability of the first
2 events occurring is $2/n$ if $x_j=x_k$ and $0$ otherwise. Thus for
any $i \in [t], j \in [n], k \in [n]$, we have that $\Pr[(A_i, B_i) =
  (j,k)]$ is either $0$ or $2/n^2$. This means that we have an
instance which is at distance $1$ from the product distribution
$[n]\times [n]$ (and hence $\Omega(1)$ from any distribution).
\end{proof}

\subsection*{Acknowledgements}
We thank Devanshi Nishit Vyas for her contribution to some of the
initial work which led to this paper.  We thank Clement Canonne for
invaluable comments on an early draft of the manuscript.  We thank
Yuval Ishai for helpful discussions.

Work supported in part by Simons Foundation (\#491119), NSF grants
CCF-1617955 and CCF-1740833.

\bibliographystyle{alpha}
\bibliography{../main,../bibfile}

\appendix

\section{Lower bound for Exact GHD for two-way communication protocols}
\label{sec:gapHamming}

We prove Lemma~\ref{lem:gapHamming} here. We first prove the following
claim.

\begin{claim}
\label{lem:gapHammingLog}
Let $n\ge 1$ be even. Let $\beta=\beta(n)=\sqrt{n/2}/4$, and
$\gamma=O(\sqrt{\log n})$.  Consider a two-way communication protocol
$\cA$ that, with probability at least $0.9$, for $x,y\in\{0,1\}^n$
with $\|x\|_1=\|y\|_1=n/2$, can distinguish between the case when
$\|x-y\|_1=n/2$ versus $\|x-y\|_1-n/2\in [\beta,\gamma\beta]$. Then
$\cA$ must exchange at least $\Omega(n/\log n)$ bits of communication.
\end{claim}

\begin{proof}
We will first prove a lower bound for a related distributional
problem, and then show a reduction from this distributional
problem. Consider the distributional problem where, for random $x,y\in
\{0,1\}^{n/2}$, a deterministic communication protocol ${\cal D}$
satisfies the following:
\begin{itemize}
\item
if $\|x-y\|_1=n/4$, the protocol outputs $-1$ with probability at least
$1-1/100\sqrt{n}$;
\item
if $\left|\|x-y\|_1- n/4\right|\ge \beta(n)/2$, the
protocol outputs $+1$ with probability at least $1-1/100\sqrt{n}$.
\end{itemize}

We will now show that any such ${\cal D}$ must use at least
$\Omega(n)$ bits.  We use the results of
\cite{DBLP:journals/toc/Sherstov12}.  Let $\mu$ be the (joint)
distribution on $x,y$ which are uniformly random from
$\{0,1\}^{n/2}$. Now define $f_{n/2}(x,y)$ to be the partial function
solved by the above assumed protocol $\cal D$: $-1$ if
$\|x-y\|_1=n/4$, and $+1$ if
$\left|\|x-y\|_1-n/4\right|\ge\beta(n)/2$. Theorem 3.3 from
\cite{DBLP:journals/toc/Sherstov12} guarantees that, for any
combinatorial rectange $R$ with $\mu(R)\ge 2^{-\delta n}$, for small
$\delta>0$, we have that $\Pr_{(x,y)\in
  R}[\left|\|x-y\|_1-n/4\right|\ge\beta]\ge 1/4$, for
$\beta(n)=\sqrt{n/2}/4$. We now use the corruption bound (Theorem 2.3
from \cite{DBLP:journals/toc/Sherstov12}). In particular, we have that
$\mu(R\cap f^{-1}_{n/2}(+1))\ge 1/4\ge 1/4\cdot \mu(R)\cdot \mu(R\cap
f^{-1}_{n/2}(-1))$. Also, $\mu(f^{-1}_{n/2}(-1))\approx
\frac{\sqrt{2}}{\sqrt{\pi n/2}}$. Hence we conclude that the
communication complexity of $\cal D$ is at least
$$ \delta n+\log
\left(\mu(f^{-1}_{n/2}(-1))-\tfrac{1/100\sqrt{n}}{1/4}\right)=\Omega(n),
$$ 
where we used the fact that the failure probability of the protocol
is $\xi=1/100\sqrt{n}\ll \mu(f^{-1}_{n/2}(-1))$.

It remains to show that the claimed $\cA$ exists, then there's also a
protocol ${\cal D}$. In particular, suppose $\cA$ exists. Then
consider input strings $x,y\in\{0,1\}^{n/2}$ to $\cal D$, chosen from
the uniform distributon $\mu$. Consider string $x'$ formed as $x$,
followed by the negation of $x$; note that $\|x'\|=n/2$. Similarly,
construct $y'$. 
Note that $\|x'\|_1=\|y'\|_1=n/2$, and that $\|x'-y'\|_1=2\cdot
\|x-y\|_1$. With probability at least $1-1/n$, we also have that
$\|x'-y'\|_1=2\|x-y\|_1\in [n/2-O(\sqrt{n\log n}),n/2+O(\sqrt{n\log
    n})]$, i.e., $\|x'-y'\|_1-n/2\in [-\gamma\beta(n),
  \gamma\beta(n)]$.  Hence Alice and Bob can run the protocol $\cA$,
on inputs $x',y'$, to distinguish whether $\|x'-y'\|_1=n/2$ or
$\|x'-y'\|_1-n/2\in [\beta(n),\gamma \beta(n)]$; note that this is
precisely equivalent to distinguishing between $\|x-y\|_1=n/4$ versus
$\|x-y\|_1-n/4\in [\beta(n)/2,\gamma\beta(n)/2]$.

This is not enough though, as $f_{n/2}(x,y)=+1$ also when
$\|x-y\|_1\le n/4-\beta(n)/2$. Thus, Alice and Bob will instead run
the protocol $\cA$ on $x',y'$, as well as on inputs $x'$ and negation
of $y'$. Thus, if $\|x-y\|_1=n/4$, then $\cA$ will return $-1$ both
times, with probability at least $0.8$. If
$\left|\|x-y\|_1-n/4\right|\in[\beta/2,\gamma\beta/2]$, then $\cA$
will return $+1$ at least once, with probability at least 0.9. Hence,
we say $\|x-y\|_1=n/4$ iff $\cA$ returns $-1$ both times.

The above reduction gives a 0.8 probability of success. We can
amplify this to $1-o(1/\sqrt{n})$ by running $O(\log n)$ independent
copies of the (randomized) protocol reduction from above, and taking
the median answer. This way we obtain a randomized protocol with
success probability at least $1-o(1/\sqrt{n})$ in solving problem
$f_{n/2}$. We can further extract a deterministic protocol ${\cal D}$
that also achieves a success probability of $1-o(1/\sqrt{n})$.

Overall, we conclude that since $\cal D$ has $\Omega(n)$
communication, the protocol $\cA$ must have $\Omega(\tfrac{n}{\log
  n})$ communication, since the reduction uses $2\cdot O(\log n)$
copies of $\cA$.

\end{proof}

We are now ready to prove Lemma~\ref{lem:gapHamming}.

\begin{proof}[Proof of Lemma~\ref{lem:gapHamming}]
Assume the contrapositive: that for each interval $[l,r]$ with
$r/l=2$ and $l\in [\Theta(\sqrt{n}),\Theta(\sqrt{n\log n})]$, there exists a protocol $A_{[l,r]}$ which can distinguish
$\|x-y\|_1=n/2$ versus $\|x-y\|_1-n/2\in [l,r]$ with constant
  probability, using only $\le C$ communication bits. For each such
  protocol, we can obtain boosted protocol $A'_{[l,r]}$ whose failure
  probability is $o(1/\log\log n)$ and communication complexity is
  $O(C\log\log\log n)$ --- as usual, by running $O(\log\log\log n)$
  independent copies of the protocol and taking the
  majority answer.

Now, in Lemma~\ref{lem:gapHammingLog}, we have the ``far'' interval
$I=[\sqrt{n/2}/4,\sqrt{n/2}/4\cdot \Theta(\sqrt{\log n})]$. We
partition this interval into $q=O(\log\log n)$ intervals $I_1,\ldots
I_q$, where each interval $I_i=[l_i,r_i]$ has $r_i/l_i\le 2$. For each
such interval, we have (by the above) a boosted protocol $A'_{I_i}$
for $i\in[q]$.

We now construct a protocol $\cA$ able to distinguish $\|x-y\|_1=n/2$
versus $\|x-y\|_1-n/2\in I$. In particular, we run all protocols
$A'_{I_1},\ldots A'_{I_q}$ on the same inputs. If $\|x-y\|_1=n/2$,
then all of them will return $-1$, except with probability $\le q\cdot
o(1/\log\log n)=o(1)$. On the other hand, if $\|x-y\|_1-n/2\in I$,
then for $i\in[q]$ such that $\|x-y\|_1-n/2\in I_i$, the algorithm
$A'_{I_i}$ will return $+1$, except with probability $o(1/\log\log
n)$. In other words, in the far case, at least one of the algorithms
will return $+1$, with probability $1-o(1)$. Hence we can distinguish
these two cases with probability $1-o(1)$.

The overall communication of $\cA$ is $O(C\log\log n\cdot \log\log\log
n)$. By Lemma~\ref{lem:gapHammingLog}, we conclude that
$C=\Omega(\tfrac{n}{\log n\log\log n\log\log\log n})$.
\end{proof}

\end{document}